\PassOptionsToPackage{unicode}{hyperref}
\PassOptionsToPackage{hyphens}{url}
\PassOptionsToPackage{dvipsnames,svgnames,x11names}{xcolor}
\documentclass[
  12pt]{article}

\usepackage{amsmath,amssymb}
\usepackage{iftex}
\ifPDFTeX
  \usepackage[T1]{fontenc}
  \usepackage[utf8]{inputenc}
  \usepackage{textcomp} % provide euro and other symbols
\else % if luatex or xetex
  \usepackage{unicode-math}
  \defaultfontfeatures{Scale=MatchLowercase}
  \defaultfontfeatures[\rmfamily]{Ligatures=TeX,Scale=1}
\fi
\usepackage{lmodern}
\ifPDFTeX\else  
\fi
\IfFileExists{upquote.sty}{\usepackage{upquote}}{}
\IfFileExists{microtype.sty}{% use microtype if available
  \usepackage[]{microtype}
  \UseMicrotypeSet[protrusion]{basicmath} % disable protrusion for tt fonts
}{}
\makeatletter
\@ifundefined{KOMAClassName}{% if non-KOMA class
  \IfFileExists{parskip.sty}{%
    \usepackage{parskip}
  }{% else
    \setlength{\parindent}{0pt}
    \setlength{\parskip}{6pt plus 2pt minus 1pt}}
}{% if KOMA class
  \KOMAoptions{parskip=half}}
\makeatother
\usepackage{xcolor}
\makeatletter
\ifx\paragraph\undefined\else
  \let\oldparagraph\paragraph
  \renewcommand{\paragraph}{
    \@ifstar
      \xxxParagraphStar
      \xxxParagraphNoStar
  }
  \newcommand{\xxxParagraphStar}[1]{\oldparagraph*{#1}\mbox{}}
  \newcommand{\xxxParagraphNoStar}[1]{\oldparagraph{#1}\mbox{}}
\fi
\ifx\subparagraph\undefined\else
  \let\oldsubparagraph\subparagraph
  \renewcommand{\subparagraph}{
    \@ifstar
      \xxxSubParagraphStar
      \xxxSubParagraphNoStar
  }
  \newcommand{\xxxSubParagraphStar}[1]{\oldsubparagraph*{#1}\mbox{}}
  \newcommand{\xxxSubParagraphNoStar}[1]{\oldsubparagraph{#1}\mbox{}}
\fi
\makeatother

\usepackage{longtable,booktabs,array}
\usepackage{calc} % for calculating minipage widths
\usepackage{etoolbox}
\makeatletter
\patchcmd\longtable{\par}{\if@noskipsec\mbox{}\fi\par}{}{}
\makeatother
\IfFileExists{footnotehyper.sty}{\usepackage{footnotehyper}}{\usepackage{footnote}}
\makesavenoteenv{longtable}
\usepackage{graphicx}
\makeatletter
\def\maxwidth{\ifdim\Gin@nat@width>\linewidth\linewidth\else\Gin@nat@width\fi}
\def\maxheight{\ifdim\Gin@nat@height>\textheight\textheight\else\Gin@nat@height\fi}
\makeatother
\setkeys{Gin}{width=\maxwidth,height=\maxheight,keepaspectratio}
\makeatletter
\def\fps@figure{htbp}
\makeatother

\makeatletter
\@ifpackageloaded{caption}{}{\usepackage{caption}}
\AtBeginDocument{%
\ifdefined\contentsname
  \renewcommand*\contentsname{Table of contents}
\else
  \newcommand\contentsname{Table of contents}
\fi
\ifdefined\listfigurename
  \renewcommand*\listfigurename{List of Figures}
\else
  \newcommand\listfigurename{List of Figures}
\fi
\ifdefined\listtablename
  \renewcommand*\listtablename{List of Tables}
\else
  \newcommand\listtablename{List of Tables}
\fi
\ifdefined\figurename
  \renewcommand*\figurename{Figure}
\else
  \newcommand\figurename{Figure}
\fi
\ifdefined\tablename
  \renewcommand*\tablename{Table}
\else
  \newcommand\tablename{Table}
\fi
}
\@ifpackageloaded{float}{}{\usepackage{float}}
\floatstyle{ruled}
\@ifundefined{c@chapter}{\newfloat{codelisting}{h}{lop}}{\newfloat{codelisting}{h}{lop}[chapter]}
\floatname{codelisting}{Listing}

\makeatother
\makeatletter
\@ifpackageloaded{caption}{}{\usepackage{caption}}
\@ifpackageloaded{subcaption}{}{\usepackage{subcaption}}
\makeatother

\ifLuaTeX
  \usepackage{selnolig}  % disable illegal ligatures
\fi
\usepackage[]{natbib}
\usepackage{bookmark}

\IfFileExists{xurl.sty}{\usepackage{xurl}}{} % add URL line breaks if available
\hypersetup{
  pdftitle={Title},
  pdfauthor={Author 1; Author 2},
  pdfkeywords={3 to 6 keywords, that do not appear in the title},
  colorlinks=true,
  linkcolor={blue},
  filecolor={Maroon},
  citecolor={Blue},
  urlcolor={Blue},
  pdfcreator={LaTeX via pandoc}}

\newcommand{\anon}{1}

\usepackage{amsmath,amsthm} 
\usepackage{enumitem}
\usepackage{ctable}
\usepackage[table,xcdraw]{xcolor}
\usepackage{colortbl, booktabs}
\usepackage{multirow}
\usepackage{hhline}

\newtheorem{proposition}{Proposition}
\usepackage[table]{xcolor}

\allowdisplaybreaks

\def\ba{{\boldsymbol{a}}}
\def\bb{{\boldsymbol{b}}}

\def\bt{{\boldsymbol{t}}}

\def\bx{{\boldsymbol{x}}}
\def\by{{\boldsymbol{y}}}

\def\bB{{\mathbf{B}}}

\def\bR{{\mathbf{R}}}
\def\bS{{\mathbf{S}}}
\def\bT{{\mathbf{T}}}
\def\bX{{\mathbf{X}}}
\def\bY{{\mathbf{Y}}}

\def\bbeta{\boldsymbol{\beta}}

\def\bmu{\boldsymbol{\mu}}

\def\bvartheta{\boldsymbol{\vartheta}}

\def\bDelta{\boldsymbol{\Delta}}

\def\bSigma{\boldsymbol{\Sigma}}

\def\bOmega{\boldsymbol{\Omega}}

\begin{document}

\def\spacingset#1{\renewcommand{\baselinestretch}%
{#1}\small\normalsize} \spacingset{1}

%%%%%%%%%%%%%%%%%%%%%%%%%%%%%%%%%%%%%%%%%%%%%%%%%%%%%%%%%%%%%%%%%%%%%%%%%%%%%%

\if1\anon
{
  \title{\bf A Unified Framework for Heterogeneity, Contamination, and Missing Data in Multivariate Regression}
  \author{Hung Tong \\
    Department of Mathematics, Rowan University \\
    Cristina Tortora \\
    Department of Mathematics and Statistics, San Jos\'e State University \\
    and \\
    Antonio Punzo \\
    Department of Economics and Business, University of Catania}
  \maketitle
} \fi

\if0\anon
{
  \bigskip
  \bigskip
  \bigskip
  \begin{center}
    {\LARGE\bf Title}
\end{center}
  \medskip
} \fi

\bigskip
\begin{abstract}
Missing values, atypical observations, and heterogeneity across latent groups are common sources of complexity in regression data. The contaminated Gaussian cluster-weighted model (CG-CWM) provides a natural framework for handling atypical observations, including outliers and leverage points, in model-based clustering. We extend the CG-CWM to data with missing-at-random (MAR) values in both the response and covariate spaces. The proposed model provides clustering in regression analysis while distinguishing typical observations, outliers, and good and bad leverage points. 
By treating covariates as random, the model preserves assignment dependence, allowing them to contribute directly to cluster formation. Maximum likelihood estimation is performed through an expectation-conditional maximization (ECM) algorithm that accounts for four sources of incomplete information: missing responses and covariates, unknown component memberships, and latent contamination indicators. Conditional on these indicators, the joint distribution of responses and covariates is multivariate Gaussian, yielding closed-form conditional distributions for missing values and incorporating missingness uncertainty directly into parameter updates. Thus, missing values are handled within model fitting rather than by preliminary imputation. 
The framework provides clustering, clusterwise regression, model-based treatment of MAR values, and detection of atypical observations. 
Performance is assessed through numerical studies under varying levels of contamination and missingness patterns, and a real data application.
\end{abstract}

\noindent%
{\it Keywords:} Cluster-weighted models; Contaminated Gaussian distribution; Missing at random; Model-based clustering; 
%Robust regression; 
ECM algorithm.% 3 to 6 keywords, that do not appear in the title
\vfill

\newpage
\spacingset{1.8} % DON'T change the spacing!

\section{Introduction}
\label{sec:introduction}

Finite mixture models provide a flexible probabilistic framework for density estimation, clustering, and classification, with the mixture components commonly interpreted as latent groups in a heterogeneous population \citep{titterington1985statistical, mclachlan2000finite, mcnicholas2016mixture}. However, in many applications, the variables under study have intrinsically different roles. In particular, a random vector
can often be naturally partitioned into a response vector $\bY$ and a covariate vector $\bX$, and the relationship between the two may vary across latent groups. In such settings, clustering the joint vector $(\bX, \bY)$ without explicitly accounting for the within-group dependence of $\bY$ on $\bX$ may obscure an important part of the underlying structure.

Mixtures of regression models provide a natural framework for accommodating cluster-specific relationships between responses and covariates \citep{desarbo1988maximum, fruhwirth2006finite}. An important distinction within this class concerns whether the covariates are regarded as fixed or random. Mixtures of regressions with fixed covariates assume \emph{assignment independence}, meaning that cluster membership is independent of the covariates. This assumption can be restrictive when the distribution of the covariates itself contains information about the latent groups \citep{hennig2000identifiablity}. Cluster-weighted models (CWMs; \citealp{gershenfeld_nonlinear_1997, wedel_mixture_1995}) instead treat the covariates as random and model the within-component joint density through the factorization $p(\bx, \by) = p(\by \mid \bx) \, p(\bx)$. Allowing the marginal distribution of $\bX$ to vary across components induces \emph{assignment dependence}, so that both the distribution of the covariates and the local regression relationship contribute to cluster formation. Gaussian CWMs and related formulations have therefore become useful tools for model-based clustering of regression data \citep{dang2017multivariate}. % ingrassia_local_2012,

A second complication arises because real data frequently contain
observations that depart from the bulk of the data. This issue is
particularly important in regression analysis, where atypical observations
can occur in either the response or the covariate space. Observations that
are atypical with respect to the conditional distribution of $\bY$ given
$\bX=\bx$ are commonly referred to as outliers, whereas observations that
are atypical with respect to the distribution of $\bX$ are leverage points.
The latter can be further distinguished into good leverage points, which
remain consistent with the local regression relationship, and bad leverage
points, which are atypical in both the covariate and conditional-response
spaces \citep{rousseeuw1987robust,rousseeuw1990unmasking}. The distinction
is important because these different types of observations can have
markedly different effects on regression estimation and clustering. Figure \ref{fig:leverage-outliers} shows an illustrative example of a bivariate dataset to clarify the distinction. Table \ref{tab:leverage-outliers} shows the corresponding schematization.

\begin{figure}[!t]
    \centering
\includegraphics[width=0.5\linewidth]{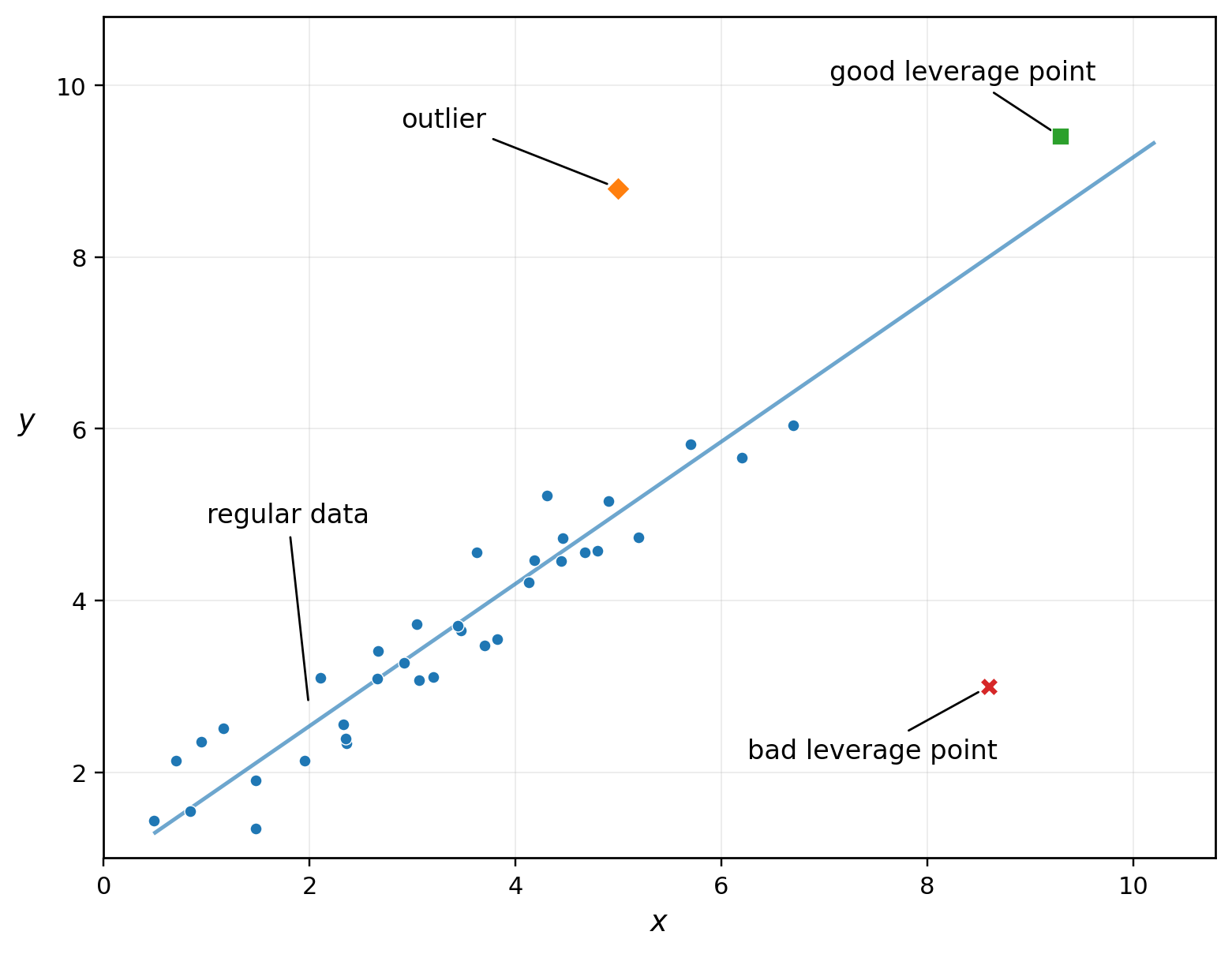}
\caption{Illustrative example of the categorization of points in a regression analysis}
\label{fig:leverage-outliers}
\end{figure}

\begin{table}[!t]
\centering
\caption{Categorization of points in a regression analysis.}
\label{tab:leverage-outliers}
\begin{tabular}{lcc}
\hline
& \multicolumn{2}{c}{Leverage (on $\bX$)} \\
Outlier (on $\bY \mid \bx$) & Yes & No \\
\hline
Yes & Bad leverage & Outlier \\
No  & Good leverage & Typical (bulk of the data) \\
\hline
\end{tabular}
\end{table}

To accommodate mildly atypical observations in this setting,
\citet{punzo_robust_2017} introduced the contaminated Gaussian
cluster-weighted model (CG-CWM). Their proposal replaces the Gaussian
distributions for both $\bX$ and $\bY\mid\bX=\bx$ with contaminated
Gaussian distributions. The contaminated Gaussian distribution
\citep{tukey_survey_1960,aitkin_mixture_1980} represents the data through
two Gaussian distributions having a common mean but different covariance
matrices, with the covariance matrix associated with the contaminating
distribution obtained by inflating that of the reference distribution.
Within each CWM component, separate contamination mechanisms for
$\bX$ and $\bY\mid\bX=\bx$ make it possible to obtain, after fitting the
model, a finer classification of observations as typical points, outliers,
good leverage points, or bad leverage points. Importantly, the proportions
and degrees of contamination are estimated from the data rather than fixed
in advance. Thus, the CG-CWM simultaneously provides model-based
clustering, clusterwise regression, and automatic detection of mildly
atypical observations.

Missing values constitute another pervasive difficulty in statistical
analysis. They occur in a wide range of applications and require particular
care because the validity of an analysis depends on the mechanism generating
the missingness \citep{rubin1976inference,little_statistical_2020}. In this
paper, we focus on the missing-at-random (MAR) mechanism, under which the
probability that a value is missing may depend on the observed data but not,
conditional on those data, on the unobserved values. Under appropriate
distinctness and parameter-space assumptions, the missing-data mechanism is
ignorable for likelihood-based inference under MAR
\citep{rubin1976inference,little_statistical_2020}. Maximum likelihood
estimation is therefore particularly attractive because it allows all the
available information to contribute to estimation without restricting the
analysis to complete observations.

The treatment of missing values becomes especially challenging in a
regression setting when missingness can occur in both $\bY$ and $\bX$.
Discarding incomplete observations can result in a substantial loss of
information and, except under more restrictive missingness conditions, can
lead to biased inference. Moreover, a preliminary single-imputation step
followed by an analysis that treats the imputed values as observed fails to
propagate the uncertainty associated with the missing entries. A likelihood
approach instead regards the missing entries as latent quantities and
integrates their uncertainty directly into model fitting, naturally leading
to the expectation-maximization (EM) framework
\citep{dempster1977maximum, mclachlan_em_2008}.

Several mixture-model approaches have been developed for incomplete data.
For model-based clustering without a response--covariate distinction,
Gaussian mixtures with missing values have been considered by
\citet{ghahramani1994supervised}, with subsequent extensions involving
heavier-tailed and more flexible component distributions
\citep{wang2004multivariate, wei2019mixtures, tong2022model, tong2024missing, pillay2026clustering}. Much less attention has been devoted to mixtures of regression models when missing values can occur simultaneously among the responses and the covariates. This distinction is consequential: with fixed covariates, missingness in $\bX$ falls outside the usual conditional regression formulation, whereas a random-covariate model provides a joint probabilistic model for both variable spaces. Recently, \citet{tong2026newlookgaussianmixtures} exploited this property to develop maximum likelihood estimation for multivariate Gaussian regression with multiple random covariates, and its model-based clustering extension, when MAR values occur in either $\bY$, $\bX$, or both.

The purpose of the present paper is to bring together the latter treatment of MAR data and the regression-clustering framework of
\citet{punzo_robust_2017}. Specifically, we introduce the contaminated Gaussian cluster-weighted model with missing-at-random values, allowing arbitrary MAR patterns in both the response and covariate spaces. This extension is not immediate. In the complete-data CG-CWM, in addition to component membership, two latent indicators identify
whether an observation belongs to the reference or contaminated distribution in $\bX$ and $\bY\mid \bX = \bx$, respectively. With incomplete data, the missing entries themselves constitute an additional source of unobserved information. Parameter estimation must therefore account
simultaneously for missing values, unknown cluster memberships, and the two latent contamination indicators.

A key result underlying our methodology is that, although the joint
distribution of $(\bX,\bY)$ within a CG-CWM component is
not itself a contaminated Gaussian distribution, it is multivariate
Gaussian conditional on the two contamination indicators. The corresponding
joint, marginal, and conditional Gaussian distributions can consequently
be exploited to derive the distributions of the missing covariates and
responses given the observed data. These results provide the conditional
moments required for likelihood-based estimation and, importantly, allow
the covariance terms associated with the uncertainty of the missing values to be retained in the estimation procedure. The conditional distribution of a missing value depends on both the component and contamination states; hence, the implicit model-based imputation adapts not only to the cluster to which an observation is likely to belong but also to whether it behaves as a typical point, an outlier, or a good/bad leverage point.

Maximum likelihood estimation is implemented using the expectation-conditional maximization (ECM) algorithm of \citet{meng93}, an extension of the EM algorithm in which the maximization step is replaced by a sequence of simpler conditional maximization steps. 
The resulting procedure retains the main advantages of the original CG-CWM: clusterwise regression and automatic identification of atypical observations, while extending them to incomplete data under MAR. 
Thus, the proposed framework simultaneously addresses four interconnected tasks: clustering heterogeneous regression data, estimating cluster-specific regression relationships, accommodating MAR values in responses and covariates, and detecting outliers and leverage points. In particular, missing-data treatment and atypical-point detection are performed jointly within a single probabilistic model rather than as separate preprocessing
and analysis stages.

The remainder of the paper is organized as follows. 
Section~\ref{sec:background} reviews the contaminated Gaussian distribution (Section~\ref{subsec:contaminated_gaussian}), the contaminated Gaussian cluster-weighted model (Section~\ref{sec:Contaminated Gaussian Cluster Weighted Models}), and the MAR framework (Section~\ref{sec:Missing Values}). 
Section~\ref{sec:CG-CWM-MAR} introduces the CG-CWM with MAR values and develops the distributional results required for maximum likelihood estimation, followed by the proposed ECM algorithm.
Section~\ref{sec:compdet} discusses operational aspects of the methodology, including
initialization, convergence, model selection, classification, and
model-based imputation. 
Section~\ref{sec:real data} illustrates the methodology on real data. 
Finally, Section~\ref{sec:conclusions} concludes the paper and discusses possible directions for future research. In addition to these main sections, Section $2$ of the Supplementary Material also evaluates the performance of the proposed approach through simulation studies.

\section{Background}
\label{sec:background}

This section provides an overview of the essential background for the paper. Specifically, we first introduce the contaminated Gaussian distribution and its use in cluster-weighted models. We then briefly discuss the treatment of missing values.

\subsection{Contaminated Gaussian Distribution}
\label{subsec:contaminated_gaussian}

Let $\bT$ be a $d$-variate random vector with real-valued components. 
$\bT$ is said to follow a contaminated Gaussian distribution (CG; \citealp{tukey_survey_1960}) with mean vector $\bmu$, scale matrix $\bSigma$, proportion of regular data $\alpha \in (0,1)$, and degree of contamination $\eta>1$, denoted by 
$\bT \sim \mathcal{CG}_d(\bmu, \bSigma, \alpha, \eta)$, if the probability density function (pdf) can be written as
\begin{equation}\label{eq:CG}
    f_{\text{CG}}(\bt;\bmu, \bSigma, \alpha, \eta)= \alpha f_{\text{G}}(\bt; \bmu, \bSigma)+(1-\alpha)f_{\text{G}}(\bt; \bmu, \eta\bSigma),
\end{equation}
where $f_{\text{G}}(\cdot; \bmu, \bSigma)$ is the pdf of a $d$-variate Gaussian distribution with mean vector $\bmu$ and covariance matrix $\bSigma$.
Thus, the CG distribution can be viewed as a two-component mixture in which one component, with probability $\alpha$, serves as the reference model for the regular observations—also referred to as ``good'' in the nomenclature of \citet{aitkin_mixture_1980}—while the other component, the contaminating one, with probability $1-\alpha$, represents anomalies (or bad observations). 
The two components share the same mean vector $\bmu$, but the component representing the anomalies has an inflated covariance matrix $\eta\bSigma$.

For identifiability, no additional constraints on the contamination parameters $\alpha$ and $\eta$ are required; see \citet{melnykov2025contaminated} and the Supplementary Material of \citet{lim2025heckmanCN}. 
However, for interpretability in robust statistics, it is sometimes assumed that $\alpha>0.5$, reflecting the idea that the majority of the observations correspond to regular data. 
In the contaminated Gaussian framework, anomalies are already defined relative to the reference distribution assumed for the regular data and are characterized by a larger variability. 
Under this conventional robust interpretation, the anomalous component is therefore regarded as a minority component. For a discussion of the notion of a reference model, see \citet{davies1993identification} and \citet{tomarchio2020dichotomous}. We impose the constraint $\alpha>0.5$ to retain this interpretation, although it is not required for identifiability and may be removed when a more flexible parameterization is desired.

Finally, for the purposes of our proposal, it is useful to note that if $\bT \sim \mathcal{CG}_d(\bmu, \bSigma, \alpha, \eta)$, then $\bT$ admits the following hierarchical representation
\begin{align}
W &\sim \mathcal{B}(\alpha) \quad \text{and} \quad \bT \mid W=w \sim \mathcal{G}_d\left(\bmu, \left(w+ \frac{1-w}{\eta}\right)^{-1}\bSigma\right), \nonumber
\end{align}
where $\mathcal{B}(\alpha)$ denotes a Bernoulli distribution with probability of success $\alpha$, and $\mathcal{G}_d(\bmu,\bSigma)$ denotes a $d$-variate Gaussian distribution with mean vector $\bmu$ and covariance matrix $\bSigma$.

% \cite{punzo_parsimonious_2016} proposed to use the CG within model-based clustering, i.e. they propose a finite mixture of $G$ CG (MCG) distributions, formally the pdf of a MCG can be written as

% \begin{equation}\label{eq:MCG}
%     f_{\text{MCG}}(\bt;\bvartheta)=\sum_{j=1}^K \pi_j f_{\text{CG}}(\bt;\bmu_j, \bSigma_j, \alpha_j, \eta_j),
% \end{equation}
% where $\pi_j \in (0,1)$ is the mixing proportion of the $g$th component with $\sum_{j=1}^K\pi_j=1$, $\bpi=\{\pi_j\}_{j=1}^K$, $\bvartheta=\{\bvartheta_j\}_{j=1}^K$, and $\bvartheta_j= \{ \bmu_j, \bSigma_j, \alpha_j, \eta_j \}$.

\subsection{Contaminated Gaussian Cluster Weighted Models}
\label{sec:Contaminated Gaussian Cluster Weighted Models}

Despite their effectiveness, the CG models mentioned previously present two main limitations. First, they do not account for potential heterogeneity in the data arising from the presence of clusters. Second, they fail to accommodate many applied settings in which the matrix $\bT$ consists of responses $\bY$ and covariates $\bX$, with dimensions $d_{\bY}$ and $d_{\bX}$, respectively. This partitioning often reflects either the intrinsic meaning of the variables or the researcher’s specific analytical objectives. Preserving this structure is therefore important for conducting an effective analysis.
In such situations, mixture models provide a natural extension, as the functional dependence of $\bY$ on $\bX=\bx$ within each mixture component can help explain and refine the clustering structure. To incorporate external information contained in random covariates, clusterwise regression with random covariates—also known as the cluster-weighted model \citep{gershenfeld_nonlinear_1997} or the saturated mixture regression model \citep{wedel_mixture_1995}—offers an important model-based clustering framework.
In this model, the joint density $p(\bx,\by;\bvartheta)$ of the pair $(\bX,\bY)$ is expressed as a mixture of $k$ components in which, within each component, the joint density factorizes into the product of the conditional density of the responses given the covariates and the marginal density of the covariates: 
\begin{align*}
    p (\bx, \by; \bvartheta) = \sum_{j = 1}^k \pi_j p (\by \mid \bx; \bvartheta_{\bY|j}) \, p (\bx; \bvartheta_{\bX|j}),
\end{align*}
 where $\pi_j \in (0,1)$ is the mixing proportion of the $j$th component such that $\sum_{j=1}^k\pi_j=1$, and $\bvartheta=\{\pi_j,\bvartheta_{\bX|j}, \bvartheta_{\bY|j}\}_{j=1}^k$. The number of clusters is assumed to be equal to the number of components $k$.
This formulation allows cluster membership to depend on the covariates $\bX$, since the distribution of $\bX$ is allowed to vary across mixture components.

\citet{punzo_robust_2017} proposed the contaminated Gaussian cluster weighted model (CG-CWM) for a random sample of $n$ covariates-response pairs $\{ \bX_i, \bY_i \}_{i = 1}^n$. Let $Z_{ij}$, for $j = 1, \dots, k$,  denote the cluster-membership indicator such that $Z_{ij} = 1$ if observation $i$ belongs to cluster $j$ and $Z_{ij} = 0$ otherwise. The CG-CWM assumes the following distributional relationships within each cluster
\begin{align*}
    \bX_i \mid Z_{ij} = 1 & \sim \mathcal{CG}_{d_\bX} \left( \bmu_{\bX \mid j}, \bSigma_{\bX \mid j}, \alpha_{\bX \mid j}, \eta_{\bX \mid j} \right) \\ 
    \text{and} \quad \bY_i \mid \bx_i, Z_{ij} = 1 & \sim \mathcal{CG}_{d_\bY} \left( \bmu_\bY (\bx_i; \bbeta_j), \bSigma_{\bY \mid j}, \alpha_{\bY \mid j}, \eta_{\bY \mid j} \right),
\end{align*}
where the local conditional mean of $\bY_i \mid \bx_i$ in the $j$th component is given by
\begin{align*}
    \underset{d_\bY \times 1}{\bmu_\bY (\bx_i; \bbeta_j)} &= \bbeta^\top_j \begin{bmatrix} 1 \\ \bx_i \end{bmatrix} =  \bb_{0j} + \bB_j \bx_i, \quad \text{with } \underset{d_\bY \times (1 + d_\bX)}{\bbeta^\top_j} = \begin{bmatrix}
        \underset{d_\bY \times 1}{\bb_{0j}}  & \underset{d_\bY \times d_\bX}{\bB_j}
    \end{bmatrix}.
\end{align*}
In this expression, $\bbeta_j$ is a $(1+d_{\bX})\times d_{\bY}$ matrix of regression coefficients consisting of the vector of intercepts $\bb_{0j}$ and the matrix of slopes $\bB_j$. Probabilistically, the joint density function of covariates $\bX_i$ and response $\bY_i$ is given by
\begin{align*}
    p_\text{CG-CWM} (\bx_i, \by_i; \bvartheta) = \sum_{j = 1}^k \pi_j f_\text{CG} \left( \by_i; \bmu_{\bY \mid j} (\bx_i; \bbeta_j), \bSigma_{\bY \mid j}, \alpha_{\bY \mid j}, \eta_{\bY \mid j} \right) f_\text{CG} \left( \bx_i; \bmu_{\bX \mid j}, \bSigma_{\bX \mid j}, \alpha_{\bX \mid j}, \eta_{\bX \mid j} \right).
\end{align*}
For a discussion on the identifiability of the model, see \citet{punzo_robust_2017}.
% For a discussion on the identifiability of the model in \eqref{eq:jdfcg-cw} see \citet{punzo_robust_2017}.
% \textcolor{magenta}{Tailor this part to the CG-CWM: \citet{dang2017multivariate} also established conditions for identifiability for G-CWR. Discussion of the identifiability of various G-CWR models can be found in \citet{hennig2000identifiablity} and \citet{ingrassia_local_2012}.}

\subsection{Missing Values}
\label{sec:Missing Values}

In a classical regression context, when missing values are present, each pair $(\bX_i, \bY_i)$ can be decomposed into observed and missing sub-vectors:  $\bX^o_i$ of dimension $d_{\bX_i}^o$, $\bX^m_i$ of dimension $d_{\bX_i}^m$, $\bY^o_i$ of dimension $d_{\bY_i}^{o}$, and $\bY^m_i$ of dimension $d_{\bY_i}^{m}$. Specifically,
\begin{align*}
    \bX_i = \begin{bmatrix}
\underset{d_{\bX_i}^o \times 1}{\bX^o_i} \\
\underset{d_{\bX_i}^m \times 1}{\bX^m_i}
\end{bmatrix}, \qquad \bY_i = \begin{bmatrix}
\underset{d_{\bY_i}^o \times 1}{\bY^o_i} \\
\underset{d_{\bY_i}^m \times 1}{\bY^m_i}
\end{bmatrix}.
\end{align*}
The superscripts $o$ and $m$ do not imply a common missingness pattern across observations; they are used in place of $o_i$ and $m_i$ for simplicity of notation. In this paper, we assume a missing at random (MAR) mechanism. Let $\bR_i$ denote the missingness indicator of dimension $(d_{\bX} + d_{\bY}) \times 1$ associated with $(\bX_i,\bY_i)$. Under the MAR mechanism, $P (\bR_i \mid \bX_i^o, \bX_i^m, \bY_i^o, \bY_i^m) = P(\bR_i \mid \bX_i^o, \bY_i^o)$. In other words, the probability of missingness may depend on the observed values $\bX_i^o$ and $\bY_i^o$, but not on the unobserved components $\bX_i^m$ and $\bY_i^m$ \citep{little_statistical_2020}.

\section{Contaminated Gaussian Cluster-Weighted Models with Missing-at-Random Values}
\label{sec:CG-CWM-MAR}

% To achieve maximum likelihood (ML) parameter estimates of $\bvartheta$, we can consider the observed likelihood function
% \begin{align*}
%     L (\bvartheta) &= \prod_{i = 1}^n\sum_{j = 1}^k \pi_j f_\text{CG} \left( \begin{bmatrix} \by^o_i \\ \by^m_i \end{bmatrix}; \bmu_{\bY \mid j} \left(\begin{bmatrix} \bx^o_i \\ \bx^m_i \end{bmatrix}; \bbeta_j\right), \bSigma_{\bY \mid j}, \alpha_{\bY \mid j}, \eta_{\bY \mid j} \right) f_\text{CG} \left( \begin{bmatrix} \bx^o_i \\ \bx^m_i \end{bmatrix}; \bmu_{\bX \mid j}, \bSigma_{\bX \mid j}, \alpha_{\bX \mid j}, \eta_{\bX \mid j} \right).
% \end{align*}

The goal of the proposed model is to describe heterogeneous data composed of responses $\bY$ and covariates $\bX$, while allowing for the presence of outliers and/or leverage points when the data contain MAR values. Maximum likelihood (ML) estimation of $\bvartheta$ can be obtained using the expectation-maximization (EM) algorithm introduced by \citet{dempster1977maximum} or one of its extensions \citep{mclachlan_em_2008}. Generally, an EM procedure concerns the complete-data log-likelihood, which accounts for observed values, missing values, and latent variables. It then iteratively alternates between an expectation (E) step and a maximization (M) step until convergence. The E-step computes the conditional expectation of the complete-data log-likelihood given the observed values and the current estimates of $\bvartheta$, while the M-step updates $\bvartheta$ to maximize that conditional expectation.

For the proposed model, we adopt the expectation-conditional maximization (ECM) algorithm of \citet{meng93}, an extension of the EM algorithm in which the M-step is replaced by a sequence of simpler conditional-maximization (CM) steps. For the proposed model, the complete-data formulation accounts for four sources of incompleteness:
\begin{itemize}
    
\item The first source is given by the missing values in the covariates and responses, $\bX^m_i$ and $\bY^m_i$, respectively, for $i = 1, \dots, n$.

\item The second source is the unknown component membership, previously defined as $Z_{ij}$, with $i = 1, \ldots, n$ and $j=1, \ldots, k$.

\item The remaining two sources arise from the fact that, for each observation, it is unknown whether the observation is a leverage point and/or an outlier with respect to component $j$; see Table~\ref{tab:leverage-outliers} for different types of data points. To denote these two sources of incompleteness, we use $U_{ij}$ and $V_{ij}$, respectively. Here, $V_{ij} = 1$ if $(\bX_i, \bY_i)$ is not a leverage point in component $j$ and $V_{ij} = 0$ otherwise. On the other hand, $U_{ij} = 1$ if $(\bX_i, \bY_i)$ is not an outlier in component $j$ and $U_{ij} = 0$ otherwise.
\end{itemize}

Accordingly, the complete-data log-likelihood is given by
% \begin{equation*}%\label{eq:cl}
% \begin{aligned}
% L_c(\boldsymbol{\vartheta})
% &=
% \prod_{i=1}^{n}\prod_{j=1}^{k}
% \Bigg\{
% \pi_j
% \left[
% \alpha_{\bX\mid j} f_\text{G} \left( \begin{bmatrix} \bx^o_i \\ \bx^m_i \end{bmatrix}; \bmu_{\bX \mid j}, \bSigma_{\bX \mid j} \right)
% \right]^{v_{ij}}
% \left[
% (1-\alpha_{\bX\mid j})
%  f_\text{G} \left( \begin{bmatrix} \bx^o_i \\ \bx^m_i \end{bmatrix}; \bmu_{\bX \mid j}, \eta_{\bX \mid j}\bSigma_{\bX \mid j} \right)
% \right]^{1-v_{ij}}
% \\
% &\times
% \left[
% \alpha_{\bY\mid j}
% f_\text{G} \left( \begin{bmatrix} \by^o_i \\ \by^m_i \end{bmatrix}; \bmu_{\bY \mid j} \left(\begin{bmatrix} \bx^o_i \\ \bx^m_i \end{bmatrix}; \bbeta_j\right), \bSigma_{\bY \mid j} \right)
% \right]^{u_{ij}}
% \left[
% (1-\alpha_{\bY\mid j})
% f_\text{G} \left( \begin{bmatrix} \by^o_i \\ \by^m_i \end{bmatrix}; \bmu_{\bY \mid j} \left(\begin{bmatrix} \bx^o_i \\ \bx^m_i \end{bmatrix}; \bbeta_j\right), \eta_{\bY \mid j} \bSigma_{\bY \mid j}\right)
% \right]^{1-u_{ij}}
% \Bigg\}^{z_{ij}} .
% \end{aligned}
% \end{equation*}
% The corresponding complete-data log-likelihood is thus
\begin{align}\label{eq:lcl}
    & l_c (\bvartheta) \\
    &= \sum_{i = 1}^n \sum_{j = 1}^k z_{ij} \log \pi_j \nonumber \\ 
    & + \sum_{i = 1}^n \sum_{j = 1}^k z_{ij} \left[ v_{ij} \log \alpha_{\bX \mid j} + (1 - v_{ij}) \log (1 - \alpha_{\bX \mid j}) + u_{ij} \log \alpha_{\bY \mid j} + (1 - u_{ij}) \log (1 - \alpha_{\bY \mid j}) \right]  \nonumber \\
    & -\dfrac{1}{2} \sum_{i = 1}^n \sum_{j = 1}^k z_{ij} \left[ \log \lvert \bSigma_{\bX \mid j} \rvert + d_\bX (1 - v_{ij}) \log \eta_{\bX \mid j} + \left( v_{ij} + \dfrac{1 - v_{ij}}{\eta_{\bX \mid j}} \right) \delta \left( \begin{bmatrix} \bx^o_i \\ \bx^m_i \end{bmatrix}, \bmu_{\bX \mid j}; \bSigma_{\bX \mid j} \right) \right] \nonumber \\
    & -\dfrac{1}{2} \sum_{i = 1}^n \sum_{j = 1}^k z_{ij} \left[ \log \lvert \bSigma_{\bY \mid j} \rvert + d_{\bY} (1 - u_{ij}) \log \eta_{\bY \mid j} + \left( u_{ij} + \dfrac{1 - u_{ij}}{\eta_{\bY \mid j}} \right) \delta \left( \begin{bmatrix} \by^o_i \\ \by^m_i \end{bmatrix}, \bmu_\bY \left( \begin{bmatrix} \bx^o_i \\ \bx^m_i \end{bmatrix}; \bbeta_j \right); \bSigma_{\bY \mid j} \right) \right], \nonumber
\end{align}
where $\delta (\ba; \bmu, \bSigma) = (\ba - \bmu)^\top \bSigma^{-1} (\ba - \bmu)$ is the squared Mahalanobis distance.
% \begin{align*}
%     \delta \left( \begin{bmatrix} \bx^o_i \\ \bx^m_i \end{bmatrix}, \bmu_{\bX \mid j}; \bSigma_{\bX \mid j} \right) = \left( \begin{bmatrix} \bx^o_i \\ \bx^m_i \end{bmatrix} - \begin{bmatrix} \bmu^o_{\bX \mid j} \\ \bmu^m_{\bX \mid j} \end{bmatrix} \right)^\top \bSigma^{-1}_{\bX \mid j} \left( \begin{bmatrix} \bx^o_i \\ \bx^m_i \end{bmatrix} - \begin{bmatrix} \bmu^o_{\bX \mid j} \\ \bmu^m_{\bX \mid j} \end{bmatrix} \right)
% \end{align*}
% and
% \begin{align*}
%     \delta \left( \begin{bmatrix} \by^o_i \\ \by^m_i \end{bmatrix}, \bmu_\bY \left( \begin{bmatrix} \bx^o_i \\ \bx^m_i \end{bmatrix}; \bbeta_j \right); \bSigma_{\bY \mid j} \right) &= \left( \begin{bmatrix} \by^o_i \\ \by^m_i \end{bmatrix} - \bbeta^\top_j \begin{bmatrix} 1 \\ \bx^o_i \\ \bx^m_i \end{bmatrix} \right)^\top \bSigma^{-1}_{\bY \mid j} \left( \begin{bmatrix} \by^o_i \\ \by^m_i \end{bmatrix} - \bbeta^\top_j \begin{bmatrix} 1 \\ \bx^o_i \\ \bx^m_i \end{bmatrix} \right).
% \end{align*}

When implementing the ECM algorithm, a major challenge arises in the E-step, where several conditional expectations must be computed. Specifically, although the random vectors $\bX_i \mid Z_{ij}=1$ and $\bY_i \mid \bx_i, Z_{ij}=1$ each follow a contaminated Gaussian distribution, their joint distribution is not contaminated Gaussian. Consequently, existing results for the contaminated Gaussian distribution do not directly yield the joint, marginal, and conditional distributions of $\bX_i^o$, $\bX_i^m$, $\bY_i^o$, and $\bY_i^m$. Therefore, before proceeding with parameter estimation, we first establish several distributional results required for the E-step. These results not only enable the closed-form computation of the required conditional expectations but are also essential for evaluating the covariance terms that account for the additional variability induced by missingness in both the covariates and responses. The following subsections establish these distributional properties, derive the covariance terms and the conditional expectations of the squared Mahalanobis distance terms required for the E-step, and then describe the ECM algorithm.

\subsection{Distributional Properties}
\label{sec:dist_res}

Considering the contaminated Gaussian random vectors $\bX_i \mid Z_{ij} = 1$ and $\bY_i \mid \bx_i, Z_{ij} = 1$, some relationships to the multivariate Gaussian distribution can be established given $V_{ij} = v_{ij}$ and $U_{ij} = u_{ij}$, where $v_{ij}, u_{ij} \in \{ 1, 0 \}$. This conditioning is important because, without conditioning on the contamination indicators, the joint distribution of $(\bX_i,\bY_i)\mid Z_{ij}=1$ is complicated. First, note that
\begin{align*}
    \bX_i \mid Z_{ij} = 1, V_{ij} = v_{ij}, U_{ij} = u_{ij} & \sim \mathcal{G}_{d_\bX} \left( \bmu_{\bX \mid j}, \left( v_{ij} + \dfrac{1 - v_{ij}}{\eta_{\bX \mid j}} \right)^{-1} \bSigma_{\bX \mid j},  \right) \\ 
    \text{and} \quad \bY_i \mid \bx_i, Z_{ij} = 1, V_{ij} = v_{ij}, U_{ij} = u_{ij} & \sim \mathcal{G}_{d_\bY} \left( \bmu_\bY (\bx_i; \bbeta_j), \left( u_{ij} + \dfrac{1 - u_{ij}}{\eta_{\bY \mid j}} \right) \bSigma_{\bY \mid j} \right).
\end{align*}
Following \citet[Ch.~2.3]{bishop2006pattern}, conditional on $Z_{ij} = 1$, $V_{ij} = v_{ij}$ and $U_{ij} = u_{ij}$, the joint distribution of $\bX_i$ and $\bY_i$ is given by
\begin{align*}
    \begin{bmatrix} \bX_i \\ \bY_i \end{bmatrix} \mid Z_{ij} = 1, V_{ij} = v_{ij}, U_{ij} = u_{ij} \sim \mathcal{G}_{d_\bX + d_\bY} \left( \begin{bmatrix} \bmu_{\bX \mid j} \\ \tilde{\bmu}_{\bY \mid j} \end{bmatrix}, \begin{bmatrix} \bSigma_{\bX \bX \mid j} & \bSigma_{\bX \bY \mid j} \\ \bSigma_{\bY \bX \mid j} & \bSigma_{\bY \bY \mid j} \end{bmatrix} \right),
\end{align*}
where
\begin{align*}
    \tilde{\bmu}_{\bY \mid j} &= \bbeta^\top_j \begin{bmatrix} 1 \\ \bmu_{\bX \mid j} \end{bmatrix} = \bb_{0j} + \bB_j \bmu_{\bX \mid j}, \quad \bSigma_{\bX \bX \mid j} = \left( v_{ij} + \dfrac{1 - v_{ij}}{\eta_{\bX \mid j}} \right)^{-1} \bSigma_{\bX \mid j}, \\
    \bSigma_{\bX \bY \mid j} &= \left( v_{ij} + \dfrac{1 - v_{ij}}{\eta_{\bX \mid j}} \right)^{-1} \bSigma_{\bX \mid j} \bB^\top_j, \quad \bSigma_{\bY \bX \mid j} = \left( v_{ij} + \dfrac{1 - v_{ij}}{\eta_{\bX \mid j}} \right)^{-1} \bB_j \bSigma_{\bX \mid j}, \\
    \text{and} \quad \bSigma_{\bY \bY \mid j} &=  \left( v_{ij} + \dfrac{1 - v_{ij}}{\eta_{\bX \mid j}} \right)^{-1} \bB_j \bSigma_{\bX \mid j} \bB^\top_j + \left( u_{ij} + \dfrac{1 - u_{ij}}{\eta_{\bY \mid j}} \right)^{-1} \bSigma_{\bY \mid j}.
\end{align*}
Since the covariates $\bX_i$ and responses $\bY_i$ can be decomposed into observed and missing sub-vectors, we have the joint distribution
\begin{align*}
    \begin{bmatrix} \bX^o_i \\ \bX^m_i \\ \bY^o_i \\ \bY^m_i \end{bmatrix} \mid Z_{ij} = 1, V_{ij} = v_{ij}, U_{ij} = u_{ij} \sim \mathcal{G}_{d_\bX + d_\bY} \left( \begin{bmatrix} \bmu^o_{\bX \mid j} \\ \bmu^m_{\bX \mid j} \\ \tilde{\bmu}^o_{\bY \mid j} \\ \tilde{\bmu}^m_{\bY \mid j} \end{bmatrix}, \begin{bmatrix}
\bSigma^{oo}_{\bX \bX \mid j} & \bSigma^{om}_{\bX \bX \mid j} & \bSigma^{oo}_{\bX \bY \mid j} & \bSigma^{om}_{\bX \bY \mid j} \\
\bSigma^{mo}_{\bX \bX \mid j} & \bSigma^{mm}_{\bX \bX \mid j} & \bSigma^{mo}_{\bX \bY \mid j} & \bSigma^{mm}_{\bX \bY \mid j} \\
\bSigma^{oo}_{\bY \bX \mid j} & \bSigma^{om}_{\bY \bX \mid j} & \bSigma^{oo}_{\bY \bY \mid j} & \bSigma^{om}_{\bY \bY \mid j} \\
\bSigma^{mo}_{\bY \bX \mid j} & \bSigma^{mm}_{\bY \bX \mid j} & \bSigma^{mo}_{\bY \bY \mid j} & \bSigma^{mm}_{\bY \bY \mid j} \\
\end{bmatrix} \right),
\end{align*}
where the parameters are also decomposed as follows
\begin{align*}
    \bmu_{\bX \mid j} &= \begin{bmatrix} \bmu^o_{\bX \mid j} \\[1ex] \bmu^m_{\bX \mid j} \end{bmatrix}, \quad \tilde{\bmu}_{\bY \mid j} =  \begin{bmatrix} \tilde{\bmu}^o_{\bY \mid j} \\[1ex] \tilde{\bmu}^m_{\bY \mid j}\end{bmatrix}, \quad \bSigma_{\bX \bX \mid j} = \begin{bmatrix} \bSigma^{oo}_{\bX \bX \mid j} & \bSigma^{om}_{\bX \bX \mid j} \\[1ex] \bSigma^{mo}_{\bX \bX \mid j} & \bSigma^{mm}_{\bX \bX \mid j} \end{bmatrix}, \quad \bSigma_{\bX \bY \mid j} = \begin{bmatrix} \bSigma^{oo}_{\bX \bY \mid j} & \bSigma^{om}_{\bX \bY \mid j} \\[1ex] \bSigma^{mo}_{\bX \bY \mid j} & \bSigma^{mm}_{\bX \bY \mid j} \end{bmatrix}, \\[2ex]
    \bSigma_{\bY \bX \mid j} &= \begin{bmatrix} \bSigma^{oo}_{\bY \bX \mid j} & \bSigma^{om}_{\bY \bX \mid j} \\[1ex] \bSigma^{mo}_{\bY \bX \mid j} & \bSigma^{mm}_{\bY \bX \mid j} \end{bmatrix}, \quad \text{and} \quad \bSigma_{\bY \bY \mid j} = \begin{bmatrix} \bSigma^{oo}_{\bY \bY \mid j} & \bSigma^{om}_{\bY \bY \mid j} \\[1ex] \bSigma^{mo}_{\bY \bY \mid j} & \bSigma^{mm}_{\bY \bY \mid j} \end{bmatrix}.
\end{align*}
This joint allows us to leverage the attractive properties of the multivariate Gaussian distribution, as detailed in \citet[Ch.~3]{MarKen79} and \citet[Ch.~4]{johnson_applied_2007}. In particular, the following distributional results can be established. 

\begin{proposition}\label{res1}
Given $Z_{ij} = 1$, $V_{ij} = v_{ij}$, and $U_{ij} = u_{ij}$, the joint distribution of observed covariates $\bX^o_i$ and observed responses $\bY^o_i$ is 
\begin{align*}
    \begin{bmatrix} \bX^o_i \\ \bY^o_i \end{bmatrix} \mid Z_{ij} = 1, V_{ij} = v_{ij}, U_{ij} = u_{ij} \sim \mathcal{G}_{d_{\bX^o_i} + d_{\bY^o_i}} \left( \begin{bmatrix} \bmu^o_{\bX \mid j} \\ \tilde{\bmu}^o_{\bY \mid j} \end{bmatrix}, \begin{bmatrix}
\bSigma^{oo}_{\bX \bX \mid j} & \bSigma^{oo}_{\bX \bY \mid j}  \\
\bSigma^{oo}_{\bY \bX \mid j} & \bSigma^{oo}_{\bY \bY \mid j}
\end{bmatrix} \right).
\end{align*}
\end{proposition}

\begin{proposition} \label{res2}
Given $Z_{ij} = 1$, observed covariates $\bx^o_i$, observed responses $\by^o_i$, $V_{ij} = v_{ij}$, and $U_{ij} = u_{ij}$, the distribution of missing covariates $\bX^m_i$ is
\begin{align*}
    \bX^m_i \mid \bx^o_i, \by^o_i, Z_{ij} = 1, V_{ij} = v_{ij}, U_{ij} = u_{ij} \sim \mathcal{G}_{d_{\bX^m_i}} \left( \bx_{ij}, \bSigma^\ast_{ij} \right),
\end{align*}
where
\begin{align*}
    \bx_{ij} &= E \left( \bX^m_i \mid \bx^o_i, \by^o_i, Z_{ij} = 1, V_{ij} = v_{ij}, U_{ij} = u_{ij} \right) \\
    &= \bmu^m_{\bX \mid j} + \begin{bmatrix} \bSigma^{mo}_{\bX \bX \mid j} & \bSigma^{mo}_{\bX \bY \mid j} \end{bmatrix} \begin{bmatrix} \bSigma^{oo}_{\bX \bX \mid j} & \bSigma^{oo}_{\bX \bY \mid j} \\ \bSigma^{oo}_{\bY \bX \mid j} & \bSigma^{oo}_{\bY \bY \mid j} \end{bmatrix}^{-1} \left( \begin{bmatrix} \bx^o_i \\ \by^o_i \end{bmatrix} - \begin{bmatrix} \bmu^o_{\bX \mid j} \\ \tilde{\bmu}^o_{\bY \mid j} \end{bmatrix} \right) \\[1ex]
    \text{and} \quad \bSigma^\ast_{ij} &= \text{Cov} \left( \bX^m_i \mid \bx^o_i, \by^o_i, Z_{ij} = 1, V_{ij} = v_{ij}, U_{ij} = u_{ij} \right) \\
    &= \bSigma^{mm}_{\bX \bX \mid j} - \begin{bmatrix} \bSigma^{mo}_{\bX \bX \mid j} & \bSigma^{mo}_{\bX \bY \mid j} \end{bmatrix} \begin{bmatrix} \bSigma^{oo}_{\bX \bX \mid j} & \bSigma^{oo}_{\bX \bY \mid j} \\ \bSigma^{oo}_{\bY \bX \mid j} & \bSigma^{oo}_{\bY \bY \mid j} \end{bmatrix}^{-1} \begin{bmatrix} \bSigma^{om}_{\bX \bX \mid j} \\ \bSigma^{om}_{\bY \bX \mid j} \end{bmatrix}.
\end{align*}
\end{proposition}

\begin{proposition}\label{res3}
Given $Z_{ij} = 1$, observed covariates $\bx^o_i$, observed responses $\by^o_i$, $V_{ij} = v_{ij}$, and $U_{ij} = u_{ij}$, the distribution of missing responses $\bY^m_i$ is
\begin{align*}
    \bY^m_i \mid \bx^o_i, \by^o_i, Z_{ij} = 1, V_{ij} = v_{ij}, U_{ij} = u_{ij} \sim \mathcal{G}_{d_{\bY^m_i}} \left( \by_{ij}, \bSigma^{\ast \ast}_{ij} \right),
\end{align*}
where
\begin{align*}
    \by_{ij} &= E \left( \bY^m_i \mid \bx^o_i, \by^o_i, Z_{ij} = 1, V_{ij} = v_{ij}, U_{ij} = u_{ij} \right) \\
    &= \tilde{\bmu}^m_{\bY \mid j} + \begin{bmatrix} \bSigma^{mo}_{\bY \bX \mid j} & \bSigma^{mo}_{\bY \bY \mid j} \end{bmatrix} \begin{bmatrix} \bSigma^{oo}_{\bX \bX \mid j} & \bSigma^{oo}_{\bX \bY \mid j} \\ \bSigma^{oo}_{\bY \bX \mid j} & \bSigma^{oo}_{\bY \bY \mid j} \end{bmatrix}^{-1} \left( \begin{bmatrix} \bx^o_i \\ \by^o_i \end{bmatrix} - \begin{bmatrix} \bmu^o_{\bX \mid j} \\ \tilde{\bmu}^o_{\bY \mid j} \end{bmatrix} \right) \\[1ex]
    \text{and} \quad \bSigma^{\ast \ast}_{ij} &= \text{Cov} \left( \bY^m_i \mid \bx^o_i, \by^o_i, Z_{ij} = 1, V_{ij} = v_{ij}, U_{ij} = u_{ij} \right) \\
    &= \bSigma^{mm}_{\bY \bY \mid j} - \begin{bmatrix} \bSigma^{mo}_{\bY \bX \mid j} & \bSigma^{mo}_{\bY \bY \mid j} \end{bmatrix} \begin{bmatrix} \bSigma^{oo}_{\bX \bX \mid j} & \bSigma^{oo}_{\bX \bY \mid j} \\ \bSigma^{oo}_{\bY \bX \mid j} & \bSigma^{oo}_{\bY \bY \mid j} \end{bmatrix}^{-1} \begin{bmatrix} \bSigma^{om}_{\bX \bY \mid j} \\ \bSigma^{om}_{\bY \bY \mid j} \end{bmatrix}.
\end{align*}
\end{proposition}

\begin{proposition}\label{res4}
Given $Z_{ij} = 1$, observed covariates $\bx^o_i$, observed responses $\by^o_i$, $V_{ij} = v_{ij}$, and $U_{ij} = u_{ij}$, the joint distribution of missing responses $\bY^m_i$ and missing covariates $\bX^m_i$ is
\begin{align*}
    \bY^m_i, \bX^m_i \mid \bx^o_i, \by^o_i, Z_{ij} = 1, V_{ij} = v_{ij}, U_{ij} = u_{ij} \sim \mathcal{G}_{d_{\bY^m_i} + d_{\bX^m_i}} \left( \bmu^{\ast \ast \ast}_{ij}, \bSigma^{\ast \ast \ast}_{ij} \right),
\end{align*}
where
\begin{align*}
    \bmu^{\ast \ast \ast}_{ij} &= \begin{bmatrix} \tilde{\bmu}^m_{\bY \mid j} \\ \bmu^m_{\bX \mid j} \end{bmatrix} + \begin{bmatrix} \bSigma^{mo}_{\bY \bY \mid j} & \bSigma^{mo}_{\bY \bX \mid j} \\ \bSigma^{mo}_{\bX \bY \mid j} & \bSigma^{mo}_{\bX \bX \mid j} \end{bmatrix} \begin{bmatrix} \bSigma^{oo}_{\bY \bY \mid j} & \bSigma^{oo}_{\bY \bX \mid j} \\ \bSigma^{oo}_{\bX \bY \mid j} & \bSigma^{oo}_{\bX \bX \mid j} \end{bmatrix}^{-1} \left( \begin{bmatrix} \bx^o_i \\ \by^o_i \end{bmatrix} - \begin{bmatrix} \bmu^o_{\bX \mid j} \\ \tilde{\bmu}^o_{\bY \mid j} \end{bmatrix} \right) \\[1ex]
    \text{and} \quad \bSigma^{\ast \ast \ast}_{ij} &= \text{Cov} \left( \bY^m_i, \bX^m_i \mid \bx^o_i, \by^o_i, Z_{ij} = 1, V_{ij} = v_{ij}, U_{ij} = u_{ij} \right) \\ & \begin{bmatrix} \bSigma^{mm}_{\bY \bY \mid j} & \bSigma^{mm}_{\bY \bX \mid j} \\ \bSigma^{mm}_{\bX \bY \mid j} & \bSigma^{mm}_{\bX \bX \mid j} \end{bmatrix} - \begin{bmatrix} \bSigma^{mo}_{\bY \bY \mid j} & \bSigma^{mo}_{\bY \bX \mid j} \\
\bSigma^{mo}_{\bX \bY \mid j} & \bSigma^{mo}_{\bX \bX \mid j} \end{bmatrix} \begin{bmatrix} \bSigma^{oo}_{\bY \bY \mid j} & \bSigma^{oo}_{\bY \bX \mid j} \\ \bSigma^{oo}_{\bX \bY \mid j} & \bSigma^{oo}_{\bX \bX \mid j} \end{bmatrix}^{-1} \begin{bmatrix} \bSigma^{om}_{\bY \bY \mid j} & \bSigma^{om}_{\bY \bX \mid j} \\ \bSigma^{om}_{\bX \bY \mid j} & \bSigma^{om}_{\bX \bX \mid j} \end{bmatrix}.
\end{align*}
\end{proposition}

\begin{proof}
The proofs for the above propositions can be found in \citet{tong2026newlookgaussianmixtures}. % It is important to note that these propositions hold only when conditioning on the contamination indicators $V_{ij}$ and $U_{ij}$.
\end{proof}

Note that the above propositions depend on the values of the latent contamination indicators $V_{ij}$ and $U_{ij}$. Consequently, quantities such as $\bx_{ij}$, $\by_{ij}$ and covariance matrices each have four possible forms. To simplify the notation, we introduce a superscript notation indicating the values of the two contamination indicators. The first symbol corresponds to $V_{ij}$, whereas the second corresponds to $U_{ij}$. A filled circle, $\bullet$, indicates that the corresponding indicator equals $1$, whereas an open circle, $\circ$, indicates that it equals $0$. For example, the conditional expectation of missing covariates $\bX^m_i$ has the following four forms:
\begin{align*}
    \overset{\bullet \bullet}{\bx}_{ij} &=
    E \left( \bX_i^m \mid \bx_i^o, \by_i^o, Z_{ij} = 1, V_{ij} = 1, U_{ij} = 1 \right), \quad \overset{\bullet \circ}{\bx}_{ij} = E \left( \bX_i^m \mid \bx_i^o, \by_i^o, Z_{ij}=1, V_{ij}=1, U_{ij}=0
    \right), \\ 
    \overset{\circ \bullet}{\bx}_{ij} &= E \left( \bX_i^m \mid \bx_i^o, \by_i^o, Z_{ij}=1, V_{ij} = 0, U_{ij} = 1 \right), \quad  \overset{\circ \circ}{\bx}_{ij}
    = E \left( \bX_i^m \mid \bx_i^o, \by_i^o, Z_{ij}=1, V_{ij} = 0, U_{ij} = 0 \right).
\end{align*}
Similarly, the four conditional expectations for the missing responses $\bY_i^m$ are denoted using the same superscript notation, giving rise to $\overset{\bullet\bullet}{\by}_{ij}$, $\overset{\bullet\circ}{\by}_{ij}$, $\overset{\circ\bullet}{\by}_{ij}$, and $\overset{\circ\circ}{\by}_{ij}$. Because these conditional expectations vary according to the values of $V_{ij}$ and $U_{ij}$, the resulting imputation depends on the type of data point described in Table~\ref{tab:leverage-outliers}.

To avoid repetition throughout the paper, we present only the case corresponding to $V_{ij} = 1$ and $U_{ij} = 1$; the remaining three cases are obtained by replacing the values of $V_{ij}$ and $U_{ij}$ accordingly. The complete notation for all quantities under the four possible combinations of $V_{ij}$ and $U_{ij}$ is provided in Section $1$ of the Supplementary Material.

\subsection{Results about Covariance and Squared Mahalanobis Distance} 
\label{sec:cov_res}

Missing values in the covariates $\bX_i^m$ and responses $\bY_i^m$ introduce additional variability, which in turn affects the maximum likelihood estimates of the model parameters. As a result, several conditional covariance matrices are required in the E-step. Recall that Propositions~\ref{res2}--\ref{res4} provide the following conditional covariance matrices:
\begin{align*}
    \bSigma^\ast_{ij} &= \text{Cov} \left( \bX_i^m \mid \bx_i^o, \by_i^o, Z_{ij} = 1, V_{ij} = v_{ij}, U_{ij} = u_{ij} \right), \\
    \bSigma^{\ast \ast}_{ij} &= \text{Cov} \left( \bY_i^m \mid \bx_i^o, \by_i^o, Z_{ij} = 1, V_{ij} = v_{ij}, U_{ij} = u_{ij} \right), \\
    \text{and } \bSigma^{\ast \ast \ast}_{ij} &= \text{Cov} \left( \bY_i^m, \bX_i^m \mid \bx_i^o, \by_i^o, Z_{ij} = 1, V_{ij} = v_{ij}, U_{ij} = u_{ij} \right).
\end{align*}
Building on these results, we have 
\begin{align*}
    \bOmega_{\bX \mid ij} &= \text{Cov} \left( \begin{bmatrix} \bx_i^o \\ \bX_i^m \end{bmatrix} \, \middle| \, \bx_i^o, \by_i^o, Z_{ij} = 1, V_{ij} = v_{ij}, U_{ij} = u_{ij} \right) = \begin{bmatrix} \boldsymbol{0} & \boldsymbol{0} \\ \boldsymbol{0} & \bSigma^\ast_{ij} \end{bmatrix}, \\[1ex]
    \bOmega_{\bX_1 \mid ij} &= \text{Cov} \left( \begin{bmatrix} 1 \\ \bx_i^o \\ \bX_i^m \end{bmatrix} \, \middle| \, \bx_i^o, \by_i^o, Z_{ij} = 1, V_{ij} = v_{ij}, U_{ij} = u_{ij} \right) = \begin{bmatrix} \boldsymbol{0} & \boldsymbol{0} \\ \boldsymbol{0} & \boldsymbol{\Omega}_{\bX \mid ij} \end{bmatrix}, \\[1ex]
    \bOmega_{\bY \mid ij} &= \text{Cov} \left( \begin{bmatrix} \by_i^o \\ \bY_i^m \end{bmatrix} \, \middle| \, \bx_i^o, \by_i^o, Z_{ij} = 1, V_{ij} = v_{ij}, U_{ij} = u_{ij} \right) = \begin{bmatrix} \boldsymbol{0} & \boldsymbol{0} \\ \boldsymbol{0} & \bSigma^{\ast \ast}_{ij} \end{bmatrix}, \\[1ex]
    \bDelta_{\bY \bX_1 \mid ij} &= \text{Cov} \left( \begin{bmatrix} \by_i^o \\ \bY_i^m \end{bmatrix}, \begin{bmatrix} 1 \\ \bx_i^o \\ \bX_i^m \end{bmatrix} \, \middle| \, \bx_i^o, \by_i^o, Z_{ij} = 1, V_{ij} = v_{ij}, U_{ij} = u_{ij} \right) = \begin{bmatrix} 0 & \boldsymbol{0} & \boldsymbol{0} \\ \boldsymbol{0} & \boldsymbol{0} & \bSigma^{\ast \ast \ast}_{ij} \end{bmatrix} \\[1ex]
    \text{and} \quad \bSigma_{\widetilde{\bY} \mid ij} &= \text{Cov} \left( \begin{bmatrix} \by^o_i \\ \bY^m_i \end{bmatrix} - \bbeta^\top_j \begin{bmatrix} 1 \\ \bx^o_i \\ \bX^m_i \end{bmatrix} \, \middle| \, \bx^o_i, \by^o_i, Z_{ij} = 1, V_{ij} = v_{ij}, U_{ij} = u_{ij} \right) \\ 
    &= \bOmega_{\bY \mid ij} - \bDelta_{\bY \bX_1 \mid ij} {\bbeta}_j - \bbeta^\top_j \bDelta^\top_{\bY \bX_1 \mid ij} + \bbeta^\top_j \bOmega_{\bX_1 \mid ij} \bbeta_j.
\end{align*}
These covariance matrices allow us to evaluate the conditional expectations of the squared Mahalanobis distances appearing in the E-step. In particular, we need
\begin{align*}
    M_{\bX \mid ij} &= E \left( \delta \left( \begin{bmatrix} \bx^o_i \\ \bX^m_i \end{bmatrix}, \bmu_{\bX \mid j}; \bSigma_{\bX \mid j} \right) \, \middle| \, \bx_i^o, \by_i^o, Z_{ij} = 1, V_{ij} = v_{ij}, U_{ij} = u_{ij}  \right) \\
    &= \text{Trace} \left( \bSigma^{-1}_{\bX \mid ij} \bOmega_{\bX \mid ij} \right) + \delta \left( \begin{bmatrix} \bx^o_i \\ \bx_{ij} \end{bmatrix}, \bmu_{\bX \mid j}, \bSigma_{\bX \mid j} \right)
\end{align*}
and
\begin{align*}
    M_{\bY \mid ij} &= E \left( \delta \left( \begin{bmatrix} \by^o_i \\ \bY^m_i \end{bmatrix}, \bmu_\bY \left( \begin{bmatrix} \bx^o_i \\ \bX^m_i \end{bmatrix}; \bbeta_j \right); \bSigma_{\bY \mid j} \right) \, \middle| \, \bx_i^o, \by_i^o, Z_{ij} = 1, V_{ij} = v_{ij}, U_{ij} = u_{ij} \right) \\
    &= \text{Trace} \left( \bSigma^{-1}_{\bY \mid ij} \bSigma_{\widetilde{\bY} \mid ij} \right) + \delta \left( \begin{bmatrix} \by^o_i \\ \by_{ij} \end{bmatrix}, \bmu_\bY \left( \begin{bmatrix} \bx^o_i \\ \bx_{ij}    
    \end{bmatrix}; \bbeta_j \right), \bSigma_{\bY \mid j} \right).
\end{align*}
As in the previous subsection, the quantities $\bOmega_{\bX \mid ij}$, $\bOmega_{\bX_1 \mid ij}$, $\bOmega_{\bY \mid ij}$, $\bDelta_{\bY \bX_1 \mid ij}$, $\bSigma_{\widetilde{\bY} \mid ij}$, $M_{\bX \mid ij}$, and $M_{\bY \mid ij}$ each have four possible forms depending on the values of $V_{ij}$ and $U_{ij}$. Accordingly, we continue to use the superscript notation using filled and open circles for simplicity. % A complete list of these quantities under all four combinations of $V_{ij}$ and $U_{ij}$ is provided in Section $1$ of the Supplementary Material. % For example, $\bOmega_{\bX \mid ij}$ has four forms: $\overset{\bullet \bullet}{\bOmega}_{\bX \mid ij}$, $\overset{\bullet \circ}{\bOmega}_{\bX \mid ij}$, $\overset{\circ \bullet}{\bOmega}_{\bX \mid ij}$, and $\overset{\circ \circ}{\bOmega}_{\bX \mid ij}$.

\subsection{E-step}

The E-step computes the conditional expectation of the complete-data log-likelihood in \eqref{eq:lcl}. To this end, the following conditional expectation terms, given the observed data $\bx^o_i$ and $\by^o_i$, need to be computed: 

\noindent
$1.\ E \left( Z_{ij} \mid \bx_i^o, \by_i^o \right)$
\qquad
$2.\ E \left( Z_{ij} V_{ij} \mid \bx_i^o, \by_i^o \right)$
\qquad
$3.\ E \left( Z_{ij} U_{ij} \mid \bx_i^o, \by_i^o \right)$ \\[1ex]
$4. \ E \left( Z_{ij} \left( V_{ij} + \dfrac{1 - V_{ij}}{\eta_{\bX \mid j}} \right) \delta \left( \begin{bmatrix} \bx^o_i \\ \bX^m_i \end{bmatrix}, \bmu_{\bX \mid j}; \bSigma_{\bX \mid j} \right)\mid \bx_i^o, \by_i^o \right)$ \\[1ex]
$5. \ E \left( Z_{ij} \left( U_{ij} + \dfrac{1 - U_{ij}}{\eta_{\bY \mid j}} \right) \delta \left( \begin{bmatrix} \by^o_i \\ \bY^m_i \end{bmatrix}, \bmu_\bY \left( \begin{bmatrix} \bx^o_i \\ \bX^m_i \end{bmatrix}; \bbeta_j \right); \bSigma_{\bY \mid j} \right)\mid \bx_i^o, \by_i^o \right)$.

% \begin{enumerate}
%     \item $E \left( Z_{ij} \mid \bx_i^o, \by_i^o \right)$; 
%     \item $E \left( Z_{ij} V_{ij} \mid \bx_i^o, \by_i^o \right)$;
%     \item $E \left( Z_{ij} U_{ij} \mid \bx_i^o, \by_i^o \right)$;
%     \item $E \left( Z_{ij} \left( V_{ij} + \dfrac{1 - V_{ij}}{\eta_{\bX \mid j}} \right) \delta \left( \begin{bmatrix} \bx^o_i \\ \bX^m_i \end{bmatrix}, \bmu_{\bX \mid j}; \bSigma_{\bX \mid j} \right)\mid \bx_i^o, \by_i^o \right)$;
%     \item $E \left( Z_{ij} \left( U_{ij} + \dfrac{1 - U_{ij}}{\eta_{\bY \mid j}} \right) \delta \left( \begin{bmatrix} \by^o_i \\ \bY^m_i \end{bmatrix}, \bmu_\bY \left( \begin{bmatrix} \bx^o_i \\ \bX^m_i \end{bmatrix}; \bbeta_j \right); \bSigma_{\bY \mid j} \right)\mid \bx_i^o, \by_i^o \right)$.
% \end{enumerate}

%%%%%%%%%%%%%%%%%%%%

Starting with the unknown component membership, we have
\begin{align*}
    \tilde{z}_{ij} = E \left( Z_{ij} \mid \bx_i^o, \by_i^o \right)
    &= P \left( Z_{ij} = 1 \mid \bx_i^o, \by_i^o \right) 
    = \dfrac{\pi_j f \left( \bx_i^o, \by_i^o \mid Z_{ij} = 1 \right)}{\sum_{j' = 1}^k \pi_{j'} f \left( \bx_i^o, \by_i^o \mid Z_{ij'} = 1 \right)}.
\end{align*}
Here, the density $f \left( \bx_i^o, \by_i^o \mid Z_{ij} = 1 \right)$ can be obtained using the law of total probability by summing over the latent contamination indicators $V_{ij}$ and $U_{ij}$. Specifically,
\begin{align*}
    f \left( \bx_i^o, \by_i^o \mid Z_{ij} = 1 \right)
    &= \sum_{v_{ij} = 0}^1 \sum_{u_{ij} = 0}^1
    f \left( \bx_i^o, \by_i^o, V_{ij} = v_{ij}, U_{ij} = u_{ij} \mid Z_{ij} = 1 \right) \\
    &= \alpha_{\bX \mid j} \alpha_{\bY \mid j}
    f \left( \bx_i^o, \by_i^o \mid Z_{ij} = 1, V_{ij} = 1, U_{ij} = 1 \right) \\
    & \quad + \alpha_{\bX \mid j} \left( 1 - \alpha_{\bY \mid j} \right) f \left( \bx_i^o, \by_i^o \mid Z_{ij} = 1, V_{ij} = 1, U_{ij} = 0
    \right) \\
    & \quad + \left(1 - \alpha_{\bX \mid j}\right) \alpha_{\bY \mid j} f \left( \bx_i^o, \by_i^o \mid Z_{ij} = 1, V_{ij} = 0, U_{ij} = 1 \right) \\ 
    & \quad + \left( 1 - \alpha_{\bX \mid j} \right)
    \left( 1 - \alpha_{\bY \mid j} \right) f \left(
    \bx_i^o, \by_i^o \mid Z_{ij} = 1, V_{ij} = 0, U_{ij} = 0 \right).
\end{align*}
In the above, each density $f \left( \bx_i^o, \by_i^o \mid Z_{ij} = 1, V_{ij} = v_{ij}, U_{ij} = u_{ij} \right)$ can be obtained using Proposition~\ref{res1}.

For the remaining expectations, the law of iterated expectations provides a useful tool. In particular, for expectation $2$, which involves the contamination indicator $V_{ij}$, we have
\begin{align*}
    E \left( Z_{ij} V_{ij} \mid \bx_i^o, \by_i^o \right) = E \left( Z_{ij} \mid \bx_i^o, \by_i^o \right) E \left( V_{ij} \mid \bx_i^o, \by_i^o, Z_{ij} = 1 \right) = \tilde{z}_{ij} \tilde{v}_{ij},
\end{align*}
where
\begin{align*}
    \tilde{v}_{ij} = E \left( V_{ij} \mid \bx_i^o, \by_i^o, Z_{ij} = 1 \right)
    &= P \left( V_{ij} = 1 \mid \bx_i^o, \by_i^o, Z_{ij} = 1 \right) = \dfrac{\alpha_{\bX \mid j} f \left( \bx_i^o, \by_i^o \mid V_{ij} = 1, Z_{ij} = 1 \right)}{f \left( \bx_i^o, \by_i^o \mid Z_{ij} = 1 \right)}.
\end{align*}
The density $f \left( \bx_i^o, \by_i^o \mid V_{ij} = 1, Z_{ij} = 1 \right)$ in the numerator can be obtained by applying the law of total probability that marginalizes over $U_{ij}$ as follows
\begin{align*}
    & f \left( \bx_i^o, \by_i^o \mid V_{ij} = 1, Z_{ij} = 1 \right) \\
    =& \, \sum_{u_{ij}=0}^1
    f \left(
    \bx_i^o, \by_i^o, V_{ij} = 1, U_{ij} = u_{ij}
    \mid Z_{ij} = 1
    \right) \\
    =& \,
    \alpha_{\bY \mid j}
    f \left(
    \bx_i^o, \by_i^o
    \mid V_{ij} = 1, U_{ij} = 1, Z_{ij} = 1
    \right) +
    \left(1-\alpha_{\bY \mid j}\right)
    f \left(
    \bx_i^o, \by_i^o
    \mid V_{ij} = 1, U_{ij} = 0, Z_{ij} = 1
    \right).
\end{align*}
Similarly, for expectation $3$ involving the contamination indicator $U_{ij}$, we obtain
\begin{align*}
    E \left( Z_{ij} U_{ij} \mid \bx_i^o, \by_i^o \right) = E \left( Z_{ij} \mid \bx_i^o, \by_i^o \right) E \left( U_{ij} \mid \bx_i^o, \by_i^o, Z_{ij} = 1 \right) = \tilde{z}_{ij} \tilde{u}_{ij},
\end{align*}
where
\begin{align*}
    \tilde{u}_{ij} &= E \left( U_{ij} \mid \bx_i^o, \by_i^o, Z_{ij} = 1 \right) =
    P \left( U_{ij} = 1 \mid \bx_i^o, \by_i^o, Z_{ij} = 1 \right) =
    \dfrac{
    \alpha_{\bY \mid j}
    f \left( \bx_i^o, \by_i^o \mid U_{ij} = 1, Z_{ij} = 1 \right)
    }{
    f \left( \bx_i^o, \by_i^o \mid Z_{ij} = 1 \right)
    }.
\end{align*}
The density $f \left( \bx_i^o, \by_i^o \mid U_{ij} = 1, Z_{ij} = 1 \right)$ in the numerator can be obtained by applying the law of total probability that marginalizes over $V_{ij}$ as follows
\begin{align*}
    & f \left( \bx_i^o, \by_i^o \mid U_{ij} = 1, Z_{ij} = 1 \right) \\
    =& \,
    \sum_{v_{ij}=0}^1
    P \left( V_{ij} = v_{ij} \mid Z_{ij} = 1 \right)
    f \left(
    \bx_i^o, \by_i^o
    \mid V_{ij} = v_{ij}, U_{ij} = 1, Z_{ij} = 1
    \right) \\
    =& \,
    \alpha_{\bX \mid j}
    f \left(
    \bx_i^o, \by_i^o
    \mid V_{ij} = 1, U_{ij} = 1, Z_{ij} = 1
    \right) +
    \left(1-\alpha_{\bX \mid j}\right)
    f \left(
    \bx_i^o, \by_i^o
    \mid V_{ij} = 0, U_{ij} = 1, Z_{ij} = 1
    \right).
\end{align*}

%%%%%%%%%%%%%%%%%%%%

Lastly, we have the expectations involving the squared Mahalanobis distances 
\begin{align*}
    & E \left( Z_{ij} \left( V_{ij} + \dfrac{1 - V_{ij}}{\eta_{\bX \mid j}} \right) \delta \left( \begin{bmatrix} \bx^o_i \\ \bX^m_i \end{bmatrix}, \bmu_{\bX \mid j}; \bSigma_{\bX \mid j} \right)\mid \bx_i^o, \by_i^o \right) \\[1ex]
    =& \, \tilde{z}_{ij} E \left( \left( V_{ij} + \dfrac{1 - V_{ij}}{\eta_{\bX \mid j}} \right) \delta \left( \begin{bmatrix} \bx^o_i \\ \bX^m_i \end{bmatrix}, \bmu_{\bX \mid j}; \bSigma_{\bX \mid j} \right) \mid \bx_i^o, \by_i^o, Z_{ij} = 1 \right) \\[1ex]
    =& \, \tilde{z}_{ij} \left[ \tilde{v}_{ij} \tilde{u}_{ij} \overset{\bullet \bullet}{M}_{\bX \mid ij} + \tilde{v}_{ij} (1 - \tilde{u}_{ij}) \overset{\bullet \circ}{M}_{\bX \mid ij} + \left( \dfrac{1 - \tilde{v}_{ij}}{\eta_{\bX \mid j}} \right) \tilde{u}_{ij} \overset{\circ \bullet}{M}_{\bX \mid ij} + \left( \dfrac{1 - \tilde{v}_{ij}}{\eta_{\bX \mid j}} \right) (1 - \tilde{u}_{ij}) \overset{\circ \circ}{M}_{\bX \mid ij} \right]
\end{align*}
and
\begin{align*}
    & E \left( Z_{ij} \left( U_{ij} + \dfrac{1 - U_{ij}}{\eta_{\bY \mid j}} \right) \delta \left( \begin{bmatrix} \by^o_i \\ \bY^m_i \end{bmatrix}, \bmu_\bY \left( \begin{bmatrix} \bx^o_i \\ \bX^m_i \end{bmatrix}; \bbeta_j \right); \bSigma_{\bY \mid j} \right)\mid \bx_i^o, \by_i^o \right) \\[1ex]
    =& \, \tilde{z}_{ij} E \left( \left( U_{ij} + \dfrac{1 - U_{ij}}{\eta_{\bY \mid j}} \right) \delta \left( \begin{bmatrix} \by^o_i \\ \bY^m_i \end{bmatrix}, \bmu_\bY \left( \begin{bmatrix} \bx^o_i \\ \bX^m_i \end{bmatrix}; \bbeta_j \right); \bSigma_{\bY \mid j} \right)\mid \bx_i^o, \by_i^o, Z_{ij} = 1 \right) \\[1ex]
    =& \, \tilde{z}_{ij} \left[ \tilde{v}_{ij} \tilde{u}_{ij} \overset{\bullet \bullet}{M}_{\bY \mid ij} + \tilde{v}_{ij} \left( \dfrac{1 -\tilde{u}_{ij}}{\eta_{\bY \mid j}} \right) \overset{\bullet \circ}{M}_{\bY \mid ij} + 1 - \tilde{v}_{ij} \tilde{u}_{ij} \overset{\circ \bullet}{M}_{\bY \mid ij} + (1 - \tilde{v}_{ij}) \left( \dfrac{1 -\tilde{u}_{ij}}{\eta_{\bY \mid j}} \right) \overset{\circ \circ}{M}_{\bY \mid ij} \right],
\end{align*}
where 
\begin{align*}
   \overset{\bullet \bullet}{M}_{\bX \mid ij} &=  \text{Trace} \left( \bSigma^{-1}_{\bX \mid ij} \overset{\bullet \bullet}{\bOmega}_{\bX \mid ij} \right) + \delta \left( \begin{bmatrix} \bx^o_i \\ \overset{\bullet \bullet}{\bx}_{ij} \end{bmatrix}, \bmu_{\bX \mid j}, \bSigma_{\bX \mid j} \right) \\
    \text{and} \quad \overset{\bullet \bullet}{M}_{\bY \mid ij} &= \text{Trace} \left( \bSigma^{-1}_{\bY \mid ij} \overset{\bullet \bullet}{\bSigma}_{\widetilde{\bY} \mid ij} \right) + \delta \left( \begin{bmatrix} \by^o_i \\ \overset{\bullet \bullet}{\by}_{ij} \end{bmatrix}, \bmu_\bY \left( \begin{bmatrix} \bx^o_i \\ \overset{\bullet \bullet}{\bx}_{ij} \end{bmatrix}; \bbeta_j \right), \bSigma_{\bY \mid j} \right).
\end{align*}
% Note that as mentioned in Sections~\ref{sec:dist_res} and~\ref{sec:cov_res}, to avoid repetition, we report only the case corresponding to $V_{ij}=1$ and $U_{ij}=1$; the other three cases are obtained by replacing the values of $V_{ij}$ and $U_{ij}$ accordingly, for detailed notation see Section $1$ of the Supplementary Material.

With the above five conditional expectations computed, the conditional expectation of the complete-data log-likelihood in~\ref{eq:lcl} is thus given by $Q (\bvartheta) = Q_1 (\pi_j, \alpha_{\bX \mid j}, \alpha_{\bY \mid j}) + Q_2 (\bmu_{\bX \mid j}, \bSigma_{\bX \mid j}, \eta_{\bX \mid j}) + Q_3 (\bbeta_j, \bSigma_{\bY \mid j}, \eta_{\bY \mid j})$, where
\begin{align*}
    Q_1 (\pi_j, \alpha_{\bX \mid j}, \alpha_{\bY \mid j}) &= \sum_{i = 1}^n \sum_{j = 1}^k \tilde{z}_{ij} \Big[ \log \pi_j + \tilde{v}_{ij} \log \alpha_{\bX \mid j} +  \tilde{u}_{ij} \log \alpha_{\bY \mid j} + \\
    & \hspace{1.0cm} (1 - \tilde{v}_{ij}) \log (1 - \alpha_{\bX \mid j}) + (1 - \tilde{u}_{ij}) \log (1 - \alpha_{\bY \mid j}) \Big], \\
    Q_2 (\bmu_{\bX \mid j}, \bSigma_{\bX \mid j}, \eta_{\bX \mid j}) &= -\dfrac{1}{2} \sum_{i = 1}^n \sum_{j = 1}^k \tilde{z}_{ij} \Bigg[ \log \lvert \bSigma_{\bX \mid j} \rvert + d_\bX (1 - \tilde{v}_{ij}) \log \eta_{\bX \mid j} + \tilde{v}_{ij} \tilde{u}_{ij} \overset{\bullet \bullet}{M}_{\bX \mid ij} + \\
    & \hspace{1.0cm} \tilde{v}_{ij} (1 - \tilde{u}_{ij}) \overset{\bullet \circ}{M}_{\bX \mid ij} + \left( \dfrac{1 - \tilde{v}_{ij}}{\eta_{\bX \mid j}} \right) \tilde{u}_{ij} \overset{\circ \bullet}{M}_{\bX \mid ij} + \left( \dfrac{1 - \tilde{v}_{ij}}{\eta_{\bX \mid j}} \right) (1 - \tilde{u}_{ij}) \overset{\circ \circ}{M}_{\bX \mid ij} \Bigg], \\
    Q_3 (\bbeta_j, \bSigma_{\bY \mid j}, \eta_{\bY \mid j}) &= -\dfrac{1}{2} \sum_{i = 1}^n \sum_{j = 1}^k \tilde{z}_{ij} \Bigg[ \log \lvert \bSigma_{\bY \mid j} \rvert + d_{\bY} (1 - \tilde{u}_{ij}) \log \eta_{\bY \mid j} + \tilde{v}_{ij} \tilde{u}_{ij} \overset{\bullet \bullet}{M}_{\bY \mid ij} +  \\
    & \hspace{1.5cm} \tilde{v}_{ij} \left( \dfrac{1 -\tilde{u}_{ij}}{\eta_{\bY \mid j}} \right) \overset{\bullet \circ}{M}_{\bY \mid ij} + (1 - \tilde{v}_{ij}) \tilde{u}_{ij} \overset{\circ \bullet}{M}_{\bY \mid ij} + (1 - \tilde{v}_{ij}) \left( \dfrac{1 -\tilde{u}_{ij}}{\eta_{\bY \mid j}} \right) \overset{\circ \circ}{M}_{\bY \mid ij} \Bigg].
\end{align*}

\subsection{M-Step}

The M-step updates are obtained by maximizing the complete-data log-likelihood in \eqref{eq:lcl} with respect to the model parameters. Since we use an ECM algorithm, the M-step is replaced by two simpler conditional-maximization (CM) steps. To simplify the notation, we suppress the iteration superscripts throughout this section.

In the first CM-step, $\eta_{\bX\mid j}$ and $\eta_{\bY\mid j}$ are held fixed, while the remaining parameters are updated. Starting with the mixing proportions $\pi_j$ and the contamination proportions $\alpha_{\bX\mid j}$ and $\alpha_{\bY\mid j}$, for $j=1,\ldots,k$, the updates are given by
\begin{align*}
    \pi_j = \dfrac{1}{n}\displaystyle\sum_{i = 1}^n \tilde{z}_{ij},
\quad
    \alpha_{\bX \mid j} = \dfrac{\displaystyle\sum_{i = 1}^n \tilde{z}_{ij} \tilde{v}_{ij}}{\displaystyle\sum_{i = 1}^n \tilde{z}_{ij}},
\quad \text{and,} \quad 
    \alpha_{\bY \mid j} = \dfrac{\displaystyle \sum_{i = 1}^n \tilde{z}_{ij} \tilde{u}_{ij}}{\displaystyle \sum_{i = 1}^n \tilde{z}_{ij}}.
\end{align*}
When the constraint on  $\alpha_{\bX \mid j}$  is imposed, any estimate of $\alpha_{\bX \mid j}$ below 0.5 is set to 0.5. The same adjustment is applied to $\alpha_{\bY \mid j}$.
The parameters relative to the independent variable $\bX$ are then updated as
\begin{align*}
    \bmu_{\bX \mid j} = \dfrac{\displaystyle \sum_{i = 1}^n \tilde{z}_{ij} \left( \tilde{v}_{ij} \tilde{u}_{ij} \begin{bmatrix} \bx^o_i \\ \overset{\bullet \bullet}{\bx}_{ij} \end{bmatrix} + \tilde{v}_{ij} (1 - \tilde{u}_{ij}) \begin{bmatrix} \bx^o_i \\ \overset{\bullet \circ}{\bx}_{ij} \end{bmatrix} + \left( \dfrac{1 - \tilde{v}_{ij}}{\eta_{\bX \mid j}} \right) \tilde{u}_{ij} \begin{bmatrix} \bx^o_i \\ \overset{\circ \bullet}{\bx}_{ij} \end{bmatrix} + \left( \dfrac{1 - \tilde{v}_{ij}}{\eta_{\bX \mid j}} \right) (1 - \tilde{u}_{ij}) \begin{bmatrix} \bx^o_i \\ \overset{\circ \circ}{\bx} _{ij} \end{bmatrix} \right)}{\displaystyle \sum_{i = 1}^n \tilde{z}_{ij} \left( \tilde{v}_{ij} + \dfrac{1 - \tilde{v}_{ij}}{\eta_{\bX \mid j}} \right)}
\end{align*}
and
\begin{align*}
    \bSigma_{\bX \mid j} = \dfrac{\displaystyle \sum_{i = 1}^n \tilde{z}_{ij} \left( \tilde{v}_{ij} \tilde{u}_{ij} \overset{\bullet \bullet}{\bS}_{\bX \mid ij} + \tilde{v}_{ij} (1 - \tilde{u}_{ij}) \overset{\bullet \circ}{\bS}_{\bX \mid ij} + \left( \dfrac{1 - \tilde{v}_{ij}}{\eta_{\bX \mid j}} \right) \tilde{u}_{ij} \overset{\circ \bullet}{\bS}_{\bX \mid ij} + \left( \dfrac{1 - \tilde{v}_{ij}}{\eta_{\bX \mid j}} \right) (1 - \tilde{u}_{ij}) \overset{\circ \circ}{\bS} _{\bX \mid ij} \right)}{\displaystyle \sum_{i = 1}^n \tilde{z}_{ij}},
\end{align*}
where
\begin{align*}
    \overset{\bullet \bullet}{\bS}_{\bX \mid ij} &= \overset{\bullet \bullet}{\bOmega}_{\bX \mid ij} + \left( \begin{bmatrix} \bx^o_i \\ \overset{\bullet \bullet}{\bx}_{ij} \end{bmatrix} - \bmu_{\bX \mid j} \right) \left( \begin{bmatrix} \bx^o_i \\ \overset{\bullet \bullet}{\bx}_{ij} \end{bmatrix} - \bmu_{\bX \mid j} \right)^\top.
\end{align*}
As in the E-step, to avoid repetition, we report only the case corresponding to $V_{ij}=1$ and $U_{ij}=1$; the other three cases are obtained by replacing the values of $V_{ij}$ and $U_{ij}$ accordingly; for detailed notation see Section $1$ of the Supplementary Material. We then update the regression coefficients
\begin{align*}
    \bbeta_j &= \left[ \displaystyle \sum_{i = 1}^n \tilde{z}_{ij} \left( \tilde{v}_{ij} \tilde{u}_{ij} \overset{\bullet \bullet}{\bB}_{\bX \mid ij} + \tilde{v}_{ij} \left( \dfrac{1 - \tilde{u}_{ij}}{\eta_{\bY \mid j}} \right) \overset{\bullet \circ}{\bB}_{\bX \mid ij} + (1 - \tilde{v}_{ij}) \tilde{u}_{ij} \overset{\circ \bullet}{\bB}_{\bX \mid ij} + (1 - \tilde{v}_{ij}) \left( \dfrac{1 - \tilde{u}_{ij}}{\eta_{\bY \mid j}} \right) \overset{\circ \circ}{\bB}_{\bX \mid ij} \right) \right]^{-1} \times \\
    & \left[ \sum_{i = 1}^n \tilde{z}_{ij} \left( \tilde{v}_{ij} \tilde{u}_{ij} \overset{\bullet \bullet}{\bB}_{\bY \bX \mid ij} + \tilde{v}_{ij} \left( \dfrac{1 - \tilde{u}_{ij}}{\eta_{\bY \mid j}} \right) \overset{\bullet \circ}{\bB}_{\bY \bX \mid ij} + (1 - \tilde{v}_{ij}) \tilde{u}_{ij} \overset{\circ \bullet}{\bB}_{\bY \bX \mid ij} + (1 - \tilde{v}_{ij}) \left( \dfrac{1 - \tilde{u}_{ij}}{\eta_{\bY \mid j}} \right) \overset{\circ \circ}{\bB}_{\bY \bX \mid ij} \right) \right]
\end{align*}
where
\begin{align*}
    \overset{\bullet \bullet}{\bB}_{\bX \mid ij} = \begin{bmatrix} 1 \\ \bx^o_i \\[1ex] \overset{\bullet \bullet}{\bx}_{ij} \end{bmatrix} \begin{bmatrix} 1 & \bx^o_i & \overset{\bullet \bullet}{\bx}_{ij} \end{bmatrix} + \overset{\bullet \bullet}{\bOmega}_{\bX_1 \mid ij} \quad \text{and} \quad
    \overset{\bullet \bullet}{\bB}_{\bY \bX \mid ij} = \begin{bmatrix} 1 \\ \bx^o_i \\[1ex] \overset{\bullet \bullet}{\bx}_{ij} \end{bmatrix} \begin{bmatrix} \by^o_i & \overset{\bullet \bullet}{\by}_{ij} \end{bmatrix} + \overset{\bullet \bullet}{\bDelta}^\top_{\bY \bX_1 \mid ij}.
\end{align*}
The last parameter updated in the first CM-step
\begin{align*}
    \bSigma_{\bY \mid j} = \dfrac{\displaystyle \sum_{i = 1}^n \tilde{z}_{ij} \left( \tilde{v}_{ij} \tilde{u}_{ij} \overset{\bullet \bullet}{\bS}_{\bY \mid ij} + \tilde{v}_{ij} \left( \dfrac{1 - \tilde{u}_{ij}}{\eta_{\bY \mid j}} \right) \overset{\bullet \circ}{\bS}_{\bY \mid ij} + (1 - \tilde{v}_{ij}) \tilde{u}_{ij} \overset{\circ \bullet}{\bS}_{\bY \mid ij} + (1 - \tilde{v}_{ij}) \left( \dfrac{1 - \tilde{u}_{ij}}{\eta_{\bY \mid j}} \right) \overset{\circ \circ}{\bS}_{\bY \mid ij} \right)}{\displaystyle \sum_{i = 1}^n \tilde{z}_{ij}},
\end{align*}
where
\begin{align*}
    \overset{\bullet \bullet}{\bS}_{\bY \mid ij} &= \overset{\bullet \bullet}{\bSigma}_{\widetilde{\bY} \mid ij} + \left( \begin{bmatrix} \by^o_i \\ \overset{\bullet \bullet}{\by}_{ij} \end{bmatrix} - \bbeta^\top_j \begin{bmatrix} 1 \\ \bx^o_i \\ \overset{\bullet \bullet}{\bx}_{ij} \end{bmatrix} \right) \left( \begin{bmatrix} \by^o_i \\ \overset{\bullet \bullet}{\by}_{ij} \end{bmatrix} - \bbeta^\top_j \begin{bmatrix} 1 \\ \bx^o_i \\ \overset{\bullet \bullet}{\bx}_{ij} \end{bmatrix} \right)^\top.
\end{align*}

In the second CM-step,  $\eta_{\bX \mid j}$ and $\eta_{\bY \mid j}$ are updated, while the other parameters are held constant.
The update for $\eta_{\bX \mid j}$ is
\begin{align*}
    \eta_{\bX \mid j} = \text{max} \left \{\eta^*, \dfrac{\displaystyle \sum_{i=1}^n \tilde{z}_{ij} \left( 1 - \tilde{v}_{ij} \right) \left( \tilde{u}_{ij} \overset{\circ \bullet}{M}_{\bX \mid ij} + (1 - \tilde{u}_{ij}) \overset{\circ \circ}{M} _{\bX \mid ij} \right)}{\displaystyle d_\bX \sum_{i=1}^n \tilde{z}_{ij} \left( 1 - \tilde{v}_{ij} \right)} \right \},
\end{align*}
where $\eta^\ast > 1$ is a small number, such as $1.001$. Similarly, the update for $\eta_{\bY \mid j}$ is
\begin{align*}
    \eta_{\bY \mid j} = \text{max} \left \{\eta^*, \dfrac{\displaystyle \sum_{i=1}^n \tilde{z}_{ij} \left( 1 - \tilde{u}_{ij} \right) \left( \tilde{v}_{ij} \overset{\bullet \bullet}{M}_{\bY \mid ij} + (1 - \tilde{v}_{ij}) \overset{\bullet \circ}{M}_{\bY \mid ij} \right)}{\displaystyle d_\bY \sum_{i=1}^n \tilde{z}_{ij} \left( 1 - \tilde{u}_{ij} \right)} \right \}.
\end{align*}

\section{Computational details}
\label{sec:compdet}

For the proposed model, we adopt the following strategy to initialize the ECM algorithm:
\begin{itemize}

\item Apply mean imputation on each variable to obtain the complete dataset $\mathcal{D}^{(0)} = \left \{ \left( \bx^{(0)}_1, \by^{(0)}_1 \right), \dots, \left( \bx^{(0)}_n, \by^{(0)}_n \right) \right \}$.

\item Perform $k$-medoids clustering \citep{kaufman_finding_1990} on $\mathcal{D}^{(0)}$ and use the obtained hard clustering solution to initialize $\tilde{z}^{(0)}_{ij}$.

\item Set the initial proportions of good observations $\alpha^{(0)}_{\bX \mid j} = \alpha^{(0)}_{\bY \mid j} = 0.6$ and degrees of contamination $\eta^{(0)}_{\bX \mid j} = \eta^{(0)}_{\bY \mid j} = 1.4$.
  % lambda <- matrix(0.1, nrow = K, ncol = p, dimnames = list(names_K, names_p))
  % alpha  <- rep(0.6, K)
  % eta    <- rep(1.4, K)

\item Initialize the mixing proportions and other parameters  as follows
\begin{align*}
    \pi^{(0)}_j &= \frac{1}{n}\sum_{i = 1}^n \tilde{z}^{(0)}_{ij}, \quad \bmu^{(0)}_{\bX \mid j} = \dfrac{\sum_{i = 1}^n \tilde{z}^{(0)}_{ij} \bx^{(0)}_i}{\sum_{i = 1}^n \tilde{z}^{(0)}_{ij}}, \quad \bSigma^{(0)}_{\bX \mid j} = \dfrac{\sum_{i = 1}^n \tilde{z}^{(0)}_{ij} \left( \bx^{(0)}_i - \bmu^{(0)}_{\bX \mid j} \right) \left( \bx^{(0)}_i - \bmu^{(0)}_{\bX \mid j} \right)^\top }{\sum_{i = 1}^n \tilde{z}^{(0)}_{ij}}, \\
    \bbeta^{(0)}_j &= \left( \sum_{i = 1}^n \tilde{z}^{(0)}_{ij} \begin{bmatrix} 1 \\ \bx^{(0)}_i \end{bmatrix} \begin{bmatrix} 1 & \bx^{(0)}_i \end{bmatrix} \right)^{-1} \left( \sum_{i = 1}^n \tilde{z}^{(0)}_{ij} \begin{bmatrix} 1 \\ \bx^{(0)}_i \end{bmatrix} \begin{bmatrix} 1 & \by^{(0)}_i \end{bmatrix} \right), \\
    \text{and} \quad \bSigma^{(0)}_{\bY \mid j} &= \dfrac{\displaystyle\sum_{i = 1}^n \tilde{z}^{(0)}_{ij} \left( \by^{(0)}_i - \bbeta^{(0)\top}_j \begin{bmatrix} 1 \\ \bx^{(0)}_i \end{bmatrix} \right) \left( \by^{(0)}_i - \bbeta^{(0)\top}_j \begin{bmatrix} 1 \\ \bx^{(0)}_i \end{bmatrix} \right)^\top }{\sum_{i = 1}^n \tilde{z}^{(0)}_{ij}}.
\end{align*}

\end{itemize}
The ECM algorithm is run until the Aitken convergence criterion is reached \citep{aitken_series_1926, mcnicholas_serial_2010}. At convergence, let $\hat{z}_{ij}$, $\hat{v}_{ij}$, and $\hat{u}_{ij}$ be the values of $\tilde{z}_{ij}$, $\tilde{v}_{ij}$, and $\tilde{u}_{ij}$, respectively. Then, the component membership of observation $(\bx_i, \by_i)$ can be determined using the maximum \emph{a posteriori} probability (MAP) operator such that $\text{MAP} (\hat{z}_{ij}) = 1$ if $\max \{ \hat{z}_{i1}, \dots, \hat{z}_{ik} \}$ occurs at $j$ and $\text{MAP} (\hat{z}_{ij}) = 0$ otherwise. When $\text{MAP} (\hat{z}_{ij}) = 1$, observation $(\bx_i, \by_i)$ is declared to belong to component $j$. To establish whether $(\bx_i, \by_i)$ is a typical observation, outlier, good leverage point, or bad leverage point in component $j$, we consider
$\hat{u}_{ij}$ and $\hat{v}_{ij}$, where $j$ is selected such that
$\text{MAP} (\hat{z}_{ij}) = 1$, and use the rule given in
Table~\ref{tab:mapping}.

\begin{table}[!t]
\centering
\caption{Rule for classifying a generic observation ($\bx_i, \by_i$) when $\text{MAP} (\tilde{z}_{ij}) = 1$ into one of the four categories of Table \ref{tab:leverage-outliers}}
\label{tab:mapping}
\begin{tabular}{lcc}
\toprule
& \multicolumn{2}{c}{$\hat{v}_{ij}$} \\
 $\hat{u}_{ij}$ & $[0,0.5)$ & $[0.5,1]$ \\
\midrule
$[0,0.5)$ & Bad leverage & Outlier \\
$[0.5,1]$  & Good leverage & Typical (bulk of the data) \\
\bottomrule
\end{tabular}
\end{table}

At convergence, let $\overset{\bullet \bullet}{\hat{\bx}}_{ij}$, $\overset{\bullet \circ}{\hat{\bx}}_{ij}$, $\overset{\circ \bullet}{\hat{\bx}}_{ij}$, and $\overset{\circ \circ}{\hat{\bx}}_{ij}$ be the conditional expectations of missing covariates. Likewise, let $\overset{\bullet \bullet}{\hat{\by}}_{ij}$, $\overset{\bullet \circ}{\hat{\by}}_{ij}$, $\overset{\circ \bullet}{\hat{\by}}_{ij}$, and $\overset{\circ \circ}{\hat{\by}}_{ij}$ be the conditional expectations of missing responses. Then, missing covariates and missing responses can be imputed by $\hat{\bx}_{ij}$ and $ \hat{\by}_{ij}$, respectively, in which for $\text{MAP} (\hat{z}_{ij}) = 1$, $\hat{\bx}_{ij} = \overset{\bullet \bullet}{\hat{\bx}}_{ij}$ and $\hat{\by}_{ij} = \overset{\bullet \bullet}{\hat{\by}}_{ij}$ if $(\bx_i, \by_i)$ is classified as a typical observation, $\hat{\bx}_{ij} = \overset{\bullet \circ}{\hat{\bx}}_{ij}$ and $\hat{\by}_{ij} = \overset{\bullet \circ}{\hat{\by}}_{ij}$ if $(\bx_i, \by_i)$ is classified as an outlier, $\hat{\bx}_{ij} = \overset{\circ \bullet}{\hat{\bx}}_{ij}$ and $\hat{\by}_{ij} = \overset{\circ \bullet}{\hat{\by}}_{ij}$ if $(\bx_i, \by_i)$ is classified as a good leverage point, and $\hat{\bx}_{ij} = \overset{\circ \circ}{\hat{\bx}}_{ij}$ and $\hat{\by}_{ij} = \overset{\circ \circ}{\hat{\by}}_{ij}$ if $(\bx_i, \by_i)$ is classified as a bad leverage point.

\section{Real Data Analysis}
\label{sec:real data}

In this section, we illustrate the proposed methodology using the well-known \textit{Automobile} dataset \citep{schlimmer1985automobile}, originally compiled from the 1985 edition of Ward's Automotive Yearbook and freely available from the UCI Machine Learning Repository. The dataset contains vehicle specifications, insurance risk, and pricing information on $n = 205$ automobiles. For our analysis, we consider all $15$ continuous variables available in the dataset. Among these, \texttt{Normalized losses} ($Y_1$) and \texttt{Price} ($Y_2$) are treated as responses, while the remaining $d_{\bX}=13$ continuous variables are treated as covariates. These covariates describe various physical and technical characteristics of the automobiles, including dimensions, curb weight, engine characteristics, horsepower, and fuel consumption. The dataset contains missing values: $41$ in \texttt{Normalized losses}, $4$ in \texttt{Price}, $4$ in \texttt{Bore}, $4$ in \texttt{Stroke}, $2$ in \texttt{Horsepower}, and $2$ in \texttt{Peak RPM}. Since the continuous variables are measured on substantially different scales, each variable is standardized by subtracting its sample mean and dividing by its sample standard deviation. The resulting standardized variables are used in the analysis.

We fit the Gaussian CWM and the proposed Contaminated Gaussian CWM, with the number of clusters $k$  varying from $1$ to $4$. Table~\ref{tab:Automobile_fitting} shows that for every value of $k$, the contaminated Gaussian CWM is a better fit than its Gaussian counterpart. The smallest BIC value of ($4560.302$) corresponds to $k = 2$.  These results support both the presence of two distinct groups of automobiles and the use of a contaminated specification to accommodate atypical observations.
The selected partition is relatively balanced, with ($111$) observations assigned to Cluster~1 and ($94$) to Cluster~2, corresponding to estimated mixing proportions of ($0.541$) and ($0.459$), respectively; see Table~\ref{auto_clusters_outliers}. The analysis identifies ($40$) good leverage points, ($21$) response outliers, and ($4$) bad leverage points, in addition to ($140$) typical observations. Cluster~1 contains a larger number of atypical observations, particularly good leverage points. Moreover, the estimated variance-inflation parameters range from ($4.167$) to ($7.869$), indicating that the atypical components are substantially more dispersed than their regular counterparts, especially in the response space. Figure~\ref{fig:auto_outliers} illustrates how the proposed model distinguishes observations with atypical covariates from those with atypical responses, with bad leverage points occurring relatively rarely.

\begin{table}[!t]
\centering
\renewcommand{\arraystretch}{0.7}
\caption{Fitting results to the normalized Automobile data.}
\label{tab:Automobile_fitting}
\begin{tabular}{llrrr}
\toprule
Model                     & $k$ & Number of parameters   & BIC        \\ \midrule
Gaussian CWM              & $1$ & $135$                         & $5764.546$  \\
                          & $2$ & $271$                         & $5128.127$  \\
                          & $3$ & $407$                         & $5264.645$  \\
                          & $4$ & $543$                         & $5447.823$  \\ \midrule
Contaminated Gaussian CWM & $1$ & $139$                         & $5661.730$  \\
                          & $2$ & $279$                         & $4560.302$ \\
                          & $3$ & $419$                         & $4995.081$ \\
                          & $4$ & $559$                         & $4681.923$ \\ \bottomrule
\end{tabular}
\end{table}

%%%%%%%%%%%%%%%%%%%%%%%%%%%%%%%%%%%%%%%%

\begin{table}[!t]
\centering
\renewcommand{\arraystretch}{0.7}
\caption{Characteristics of clusters and their typical/atypical points according to the best-fitting model.}
\label{auto_clusters_outliers}
\begin{tabular}{lrrllrr}
\toprule
                     & Cluster $1$ & Cluster $2$ &  &                       & Cluster $1$ & Cluster $2$ \\ \midrule
Cluster size         & 111         & 94          &  & $\pi_j$               & 0.541       & 0.459       \\
Bad leverage points  & 3           & 1           &  & $\alpha_{\bX \mid j}$ & 0.723       & 0.844       \\
Outliers             & 13          & 8           &  & $\alpha_{\bY \mid j}$ & 0.797       & 0.753       \\
Good leverage points & 27          & 13          &  & $\eta_{\bX \mid j}$   & 4.942       & 4.167       \\
Typical points       & 68          & 72          &  & $\eta_{\bY \mid j}$   & 7.869       & 5.043       \\ \bottomrule
\end{tabular}
\end{table}

%%%%%%%%%%%%%%%%%%%%%%%%%%%%%%%%%%%%%%%%

\begin{figure}[!t]
    \centering
    \includegraphics[width=4.35in]{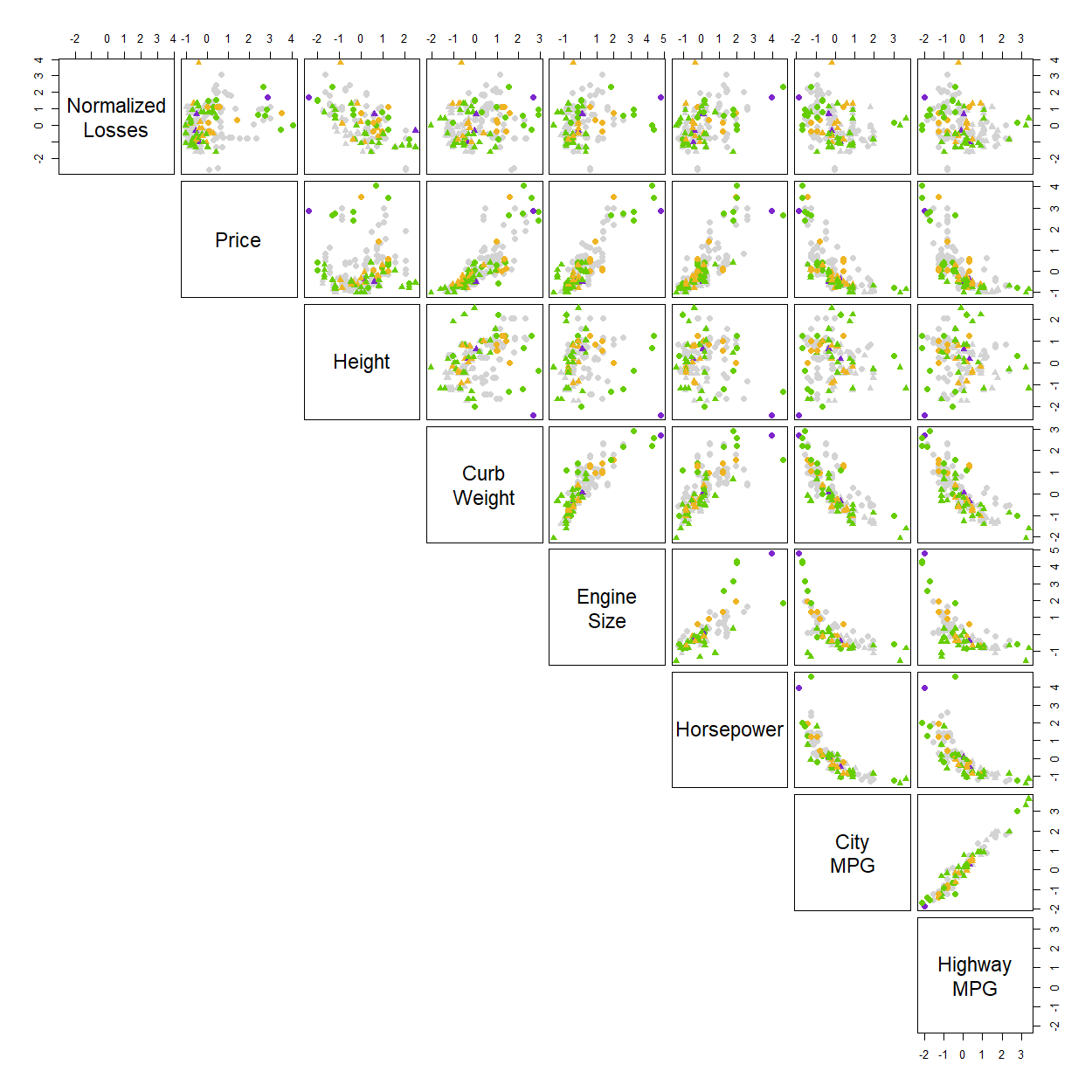}
    \caption{Different types of points indicated by the best-fitting model. Bad leverage points are purple, outliers are gold, good leverage points are green, and typical points are gray.}
    \label{fig:auto_outliers}
\end{figure}

%%%%%%%%%%%%%%%%%%%%%%%%%%%%%%%%%%%%%%%%

The scatterplots in Figure~\ref{fig:auto_y_vs_x} indicate that the separation between the clusters is more evident for \texttt{Price} and for several physical and technical covariates than for \texttt{Normalized losses}. The standardized covariate means in Figure~\ref{fig:auto_dumbbell} provide a clear interpretation of the two groups. Cluster~1 is characterized by smaller values of wheelbase, length, width, curb weight, engine size, and horsepower, together with higher city and highway fuel economy. Cluster~2 exhibits the opposite profile and generally contains larger, heavier, more powerful, and more expensive automobiles. Additional results supporting the interpretation of the two clusters are provided in Section $3$ of the Supplementary Material. %Appendix~\ref{app:auto_results}. 
These include the correlations between standardized covariates, distributions of automobile makes and symboling ratings, cluster-specific regression coefficients, and within-cluster covariance estimates.

%%%%%%%%%%%%%%%%%%%%%%%%%%%%%%%%%%%%%%%%

\begin{figure}[!t]
    \centering
    \includegraphics[width=4.0in]{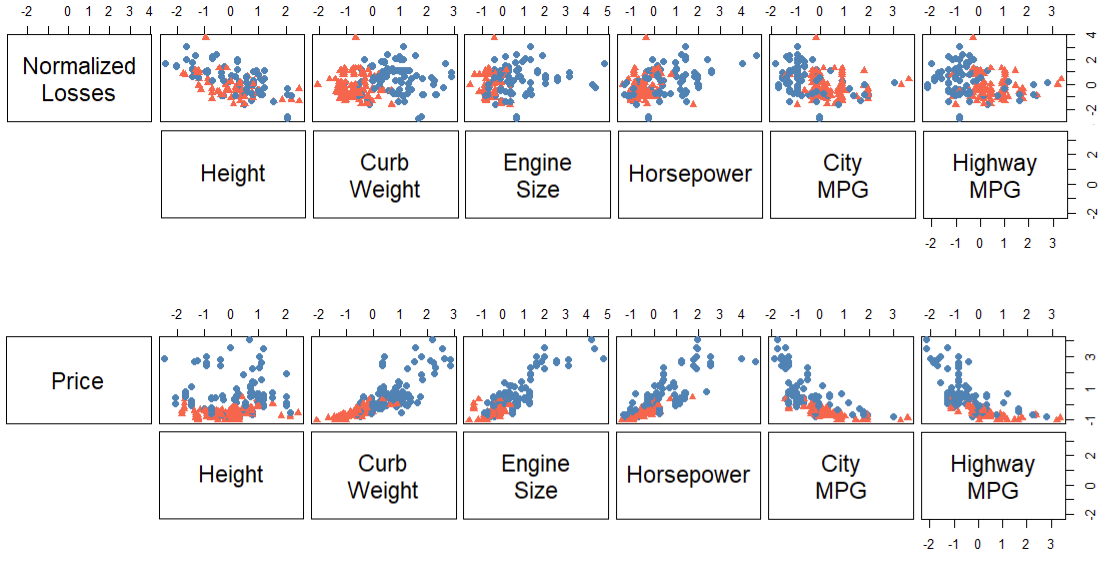} \\
    \includegraphics[width=4.50in]{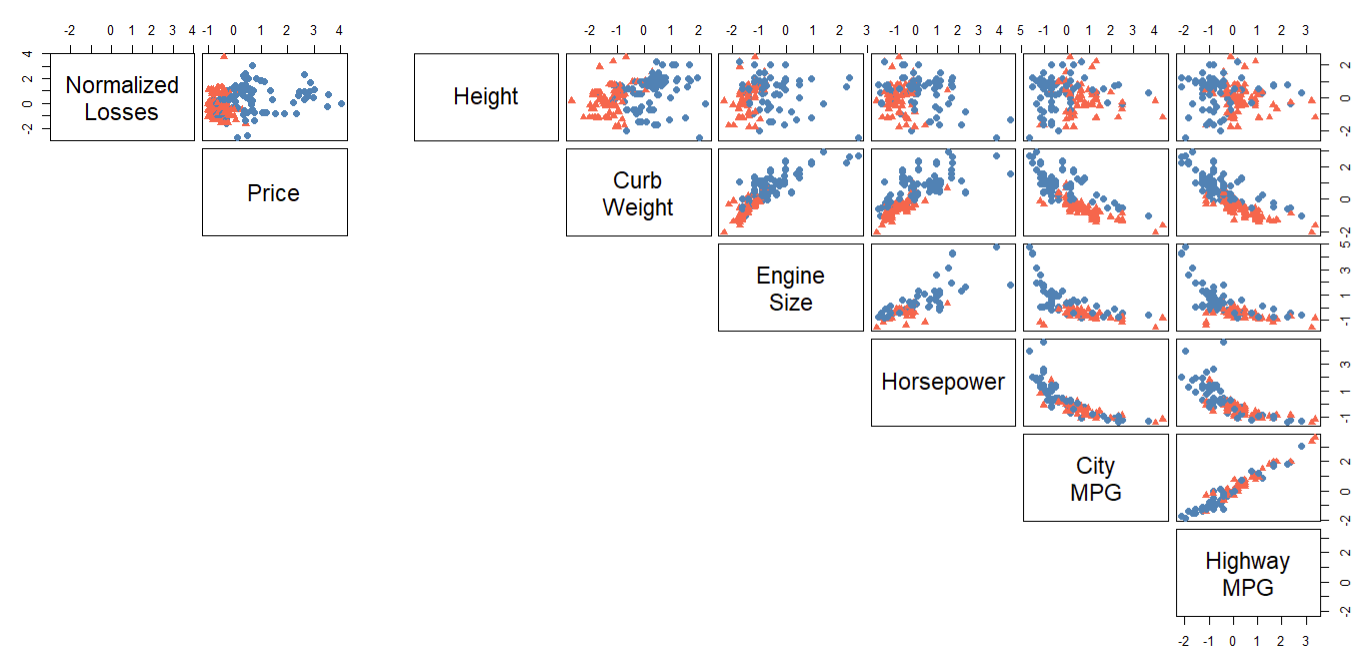}
    \caption{Subset of pairwise scatterplots of the standardized Automobile data, with missing values imputed and points colored according to the best-fitting model. The plots visualize the response-covariate, response-response, and covariate-covariate relationships.}
    \label{fig:auto_y_vs_x}
\end{figure}

%%%%%%%%%%%%%%%%%%%%%%%%%%%%%%%%%%%%%%%%

% \begin{figure}[!t]
%     \centering
    
%     \caption{Subset of pairwise scatterplots of the normalized Automobile data, with missing values imputed and points colored according to the best-fitting model. The left figure visualizes the relationship between responses, whereas the right figure visualizes the relationships between covariates.}
%     \label{fig:auto_yy_xx}
% \end{figure}

%%%%%%%%%%%%%%%%%%%%%%%%%%%%%%%%%%%%%%%%

\begin{figure}[!t]
    \centering
    \includegraphics[width=3.25in]{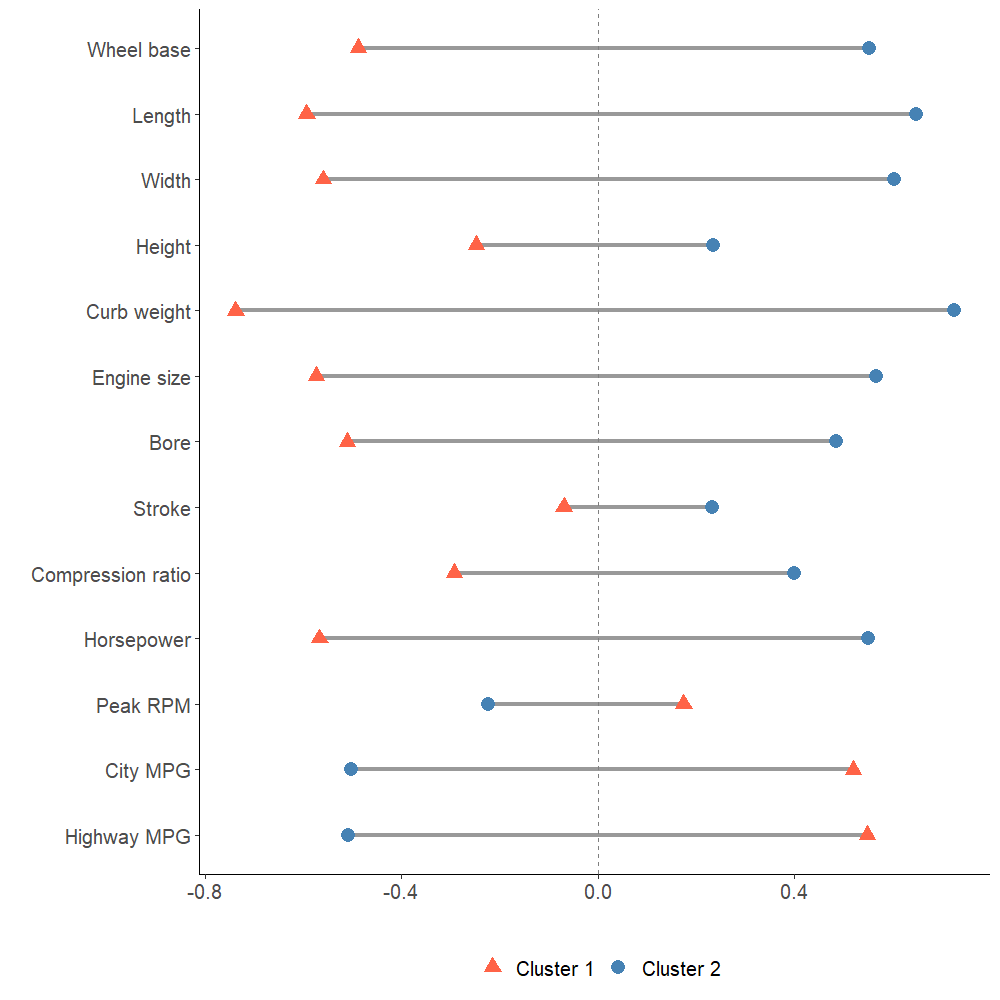}
    \caption{Comparison of standardized covariate means between the two clusters obtained by the best-fitting model.}
    \label{fig:auto_dumbbell}
\end{figure}

%%%%%%%%%%%%%%%%%%%%%%%%%%%%%%%%%%%%%%%%

% \begin{figure}[!t]
%     \centering
%     \includegraphics[width=6.15in]{fig_auto_boxplot.png}
%     \caption{Box plots of standardized variables by cluster and point type obtained by the best-fitting model.}
%     \label{fig:auto_boxplot}
% \end{figure}

\section{Conclusions}
\label{sec:conclusions}

% In this paper, we extended the contaminated Gaussian cluster-weighted
% model (CG-CWM) of \citet{punzo_robust_2017} to accommodate
% missing-at-random (MAR) values in both the response and covariate spaces.
% The resulting framework combines three features that are often required
% simultaneously in the analysis of heterogeneous regression data:
% clusterwise regression with random covariates, robustness to mildly
% atypical observations, and likelihood-based treatment of missing values.

In this paper, we propose a framework that combines three features that are often required simultaneously in the analysis of heterogeneous regression data: clusterwise regression with random covariates, robustness to mildly atypical observations, and likelihood-based treatment of missing data. This is achieved by extending the contaminated Gaussian cluster-weighted model (CG-CWM) of \citet{punzo_robust_2017} to accommodate values missing at random (MAR) in both the response and covariate spaces. In particular, the proposed model retains the ability of the CG-CWM to distinguish, within each cluster, among typical observations, outliers, good leverage points, and bad leverage points, while allowing arbitrary patterns of MAR values in both $\bY$ and $\bX$. The introduction of missing values makes parameter estimation substantially more involved than in the complete-data CG-CWM. In addition to the unknown component memberships and the two latent contamination indicators associated with $\bX$ and $\bY \mid \bX = \bx$, the
missing entries of the responses and covariates constitute an additional source of incompleteness. To address these sources jointly, we developed an ECM algorithm for maximum likelihood estimation.

A key element in the development of the proposed algorithm is the distributional representation obtained by conditioning on the two contamination indicators. Although the joint distribution of $\bX$ and $\bY$ within a CG-CWM component is not, in general, a contaminated Gaussian distribution, conditional on the contamination states it becomes multivariate Gaussian. This result makes it possible to derive the joint, marginal, and conditional distributions required to evaluate the E-step in closed form. In particular, it provides the conditional moments of the missing covariates and responses, together with the covariance terms needed to account for the uncertainty induced by missingness. An important consequence of this formulation is that the treatment of missing values is intrinsically connected with both clustering and atypical-point detection. Rather than filling in the missing entries in a preprocessing stage and subsequently treating them as observed, the proposed likelihood-based approach accounts for their uncertainty throughout model estimation. Moreover, the conditional distributions of the missing values depend on the component membership and on the latent contamination states. Hence, the information used for the implicit model-based imputation adapts to the local regression structure and to the type of observation under consideration. This feature is particularly relevant when missingness and atypicality coexist, because an observation that is likely to be an outlier or a leverage point need not be treated in the same way as a typical observation belonging to the same cluster.

% Choose/adapt the following paragraph according to the final simulation
% study included in the paper.
The numerical investigations considered in this paper illustrate the ability of the proposed methodology to recover the underlying clustering and regression structure in the presence of incomplete data, while preserving the robust properties of the contaminated Gaussian formulation. They also highlight the importance of incorporating the uncertainty associated with the missing entries directly into parameter estimation rather than relying on an external imputation procedure.

% Choose/adapt the following paragraph according to the final real-data
% application included in the paper.
The application to the \textit{Automobile} dataset complements the analysis based on the Gaussian CWM with MAR values by illustrating the additional information provided by the contaminated Gaussian formulation. Using the same response-covariate partition and a comparable analysis strategy, the proposed model simultaneously accommodates the naturally occurring missing values, identifies the latent clustering structure, and estimates the cluster-specific regression relationships. More importantly, the contaminated formulation provides a further layer of interpretation by assessing atypicality separately in the covariate space and in the conditional response space. Thus, within each cluster, automobiles can be further classified as typical observations, outliers, good leverage points, or bad leverage points. The application therefore illustrates how the proposed framework goes beyond clustering and model-based treatment of missing values, providing a unified characterization of both between-cluster heterogeneity and within-cluster atypicality in incomplete regression data.

Several extensions of the present work are possible. 
First, more flexible component distributions could be considered to accommodate skewness and other departures from the contaminated Gaussian assumption, in the fashion of the skewed cluster-weighted models of \citet{gallaugher2022multivariate}. Second, parsimonious parameterizations could be introduced to facilitate the application of the proposed model to higher-dimensional data. In particular, dimension reduction could be achieved through a factor-analytic decomposition of the component covariance matrices, along the lines of the cluster-weighted factor analyzers of \citet{subedi2013clustering,subedi2015cluster}. 
Alternatively, parsimonious covariance structures based on eigen-decomposition could be considered, following \citet{dang2017multivariate} for Gaussian cluster-weighted models and \citet{punzo2016parsimonious} for mixtures of contaminated Gaussian distributions. The latter approach could also be developed toward high-dimensional settings, in the spirit of \citet{punzo2020highdimensional}. Third, the methodology could be extended to other forms of cluster-weighted models, including models for non-continuous responses, following, for example, the generalized linear mixed cluster-weighted framework of \citet{ingrassia2015generalized}. Finally, the present analysis relies on the MAR assumption. Extending the framework to missing-not-at-random (MNAR) mechanisms would require additional assumptions on the missingness mechanism, for example through selection or pattern-mixture formulations \citep{little_statistical_2020}, and represents an interesting direction for future research.

% \blue{Overall, the proposed CG-CWM with MAR values provides a unified
% model-based framework in which clustering, regression, robust detection of
% atypical observations, and the treatment of missing responses and
% covariates are performed simultaneously. In doing so, it extends the scope
% of contaminated Gaussian cluster-weighted models to incomplete data while
% preserving their interpretation in terms of typical points, outliers, and
% good and bad leverage points.}

\section{Funding Details}\label{funding}

No funding was received for this research.

\section{Disclosure Statement}\label{disclosure-statement}

The authors have the following conflicts of interest to declare.

% The authors have the following conflicts of interest to declare (or replace with a statement that no conflicts of interest exist).

\section{Data Availability Statement}\label{data-availability-statement}

The \textit{Automobile} dataset \citep{schlimmer1985automobile} is freely accessible from the UCI Machine Learning Repository at \url{https://archive.ics.uci.edu/dataset/10/automobile}

% Deidentified data have been made available at the following URL: XX.

\phantomsection\label{supplementary-material}
\bigskip

\begin{center}

{\large\bf SUPPLEMENTARY MATERIAL}

\end{center}

\begin{description}
\item[Description:] The Supplementary Material contains a complete list of the notation used in the ECM algorithm, a simulation study to evaluate the proposed methodology, and additional results for the real data analysis.

% \item[Title:]
% Brief description. (file type)
% \item[R-package for MYNEW routine:]
% R-package MYNEW containing code to perform the diagnostic methods
% described in the article. The package also contains all datasets used as
% examples in the article. (GNU zipped tar file)
% \item[HIV data set:]
% Data set used in the illustration of MYNEW method in
% Section~\ref{sec-verify} (.txt file).
\end{description}

\bibliography{bibliography.bib}

\end{document}